\documentclass[11pt,leqno]{article}
\usepackage{amsthm,amsmath, mathtools}
\usepackage[round]{natbib}
\usepackage{amsfonts}
\usepackage{ifthen}
\usepackage{enumitem}
\usepackage{amssymb,dsfont,url,color,booktabs}
 \usepackage{tikz}
 \usepackage{setspace}

\usepackage{epstopdf}
\usepackage{chngcntr}
\usepackage{bbm}
\usepackage{array}
\usepackage[export]{adjustbox}[2011/08/13] 
\usepackage{geometry}
\usepackage{rotating}
\usepackage{todonotes}
\usepackage{xspace}
\usepackage{algorithm}
\usepackage[noend]{algpseudocode}
 \usepackage[english]{babel}
\usepackage{bbm}

\usepackage{caption}
\usepackage{subcaption}
\usepackage{multirow}
\usepackage{booktabs}
\usepackage[font={footnotesize}]{caption}

\usepackage{placeins} 		

\newcommand*{\IR}{\mathbb{R}}

\theoremstyle{plain}
\newtheorem{theorem}{Theorem}[section]

\newcommand{\notiz}[1]{\relax}

\newcommand{\zitep}[1]{\relax}

\newcommand{\1}{\mathds 1}            

\newcommand{\Price}[1][]{
		\ifthenelse{\equal{#1}{}}{\mathit{Price}}{\Price{}^{#1}}
	} 

\newlength{\wordlength}

\renewcommand{\cite}{\citet}

\numberwithin{equation}{section}
\numberwithin{figure}{section}
\numberwithin{table}{section}

\begin{document}
\title{\textbf{The Chebyshev method for the implied volatility}
}

\bigskip
\author{\textbf{Kathrin Glau$\vphantom{l}^{1}$,} \textbf{Paul Herold$\vphantom{l}^{1}$,} \textbf{Dilip B. Madan$\vphantom{l}^{2}$,} \textbf{Christian P{\"o}tz$\vphantom{l}^{1,}$\footnote{The authors would like to thank the KPMG Center of Excellence in Risk Management for their support.}
}
\\\\$\vphantom{l}^{\text{1}}$Technical University of Munich, Germany,\\
 $\vphantom{l}^{\text{2}}$ Robert H. Smith School of Business, University of Maryland
}

\maketitle
\begin{abstract}
The implied volatility is a crucial element of any financial toolbox, since it is used for quoting and the hedging of options as well as for model calibration. In contrast to the Black-Scholes formula its inverse, the implied volatility, is not explicitly available and numerical approximation is required. We propose a bivariate interpolation of the implied volatility surface based on Chebyshev polynomials. This yields a closed-form approximation of the implied volatility, which is easy to implement and to maintain. We prove a subexponential error decay. This allows us to obtain an accuracy close to machine precision with polynomials of a low degree. We compare the performance of the method in terms of runtime and accuracy to the most common reference methods. In contrast to existing interpolation methods, the proposed method is able to compute the implied volatility for all relevant option data. In this context, numerical experiments confirm a considerable increase in efficiency, especially for large data sets.
\end{abstract}

	

\doublespacing

\section{Motivation}
Ever since \cite{BlackScholes1973} and \cite{Merton1973} introduced their option pricing model, the Black-Scholes formula has been omnipresent in the financial industry. 
The one parameter in the model that can not be observed using market data is the volatility of the underlying asset process. The Black-Scholes call price function is strictly monotone increasing in volatility. Hence, for each observed call price there is a unique volatility such that the resulting model price equals the market price. This is called the \emph{implied volatility}, one of the most important quantities in finance.

The implied volatility can be seen as a universal language in the daily business of trading, hedging, model calibration and more generally in risk management.
Typically, trading desks quote option prices in implied volatilities instead of absolute prices. This allows traders to compare option prices on different underlyings such as equities, indices, currencies or commodities. For high frequency trading in particular, very accurate real-time evaluations of the implied volatility are required for large data sets. As stated in \cite{Baumeister2013} and \cite{Celis2017} in practice, often millions of option prices have to be inverted in real-time for instance by large data providers. Furthermore, the implied volatility is needed for the most common derivative hedging strategy, the so-called delta-hedging strategy. It is used to infer the sensitivity of the option price with respect to the underlying spot price, the option's delta. One takes an opposing position to the delta in the underlying asset as a hedge. Since the 1970s a large variety of asset price models that generalize and improve the Black-Scholes model have been introduced. Typically, these models are determined by a number of parameters that are fitted to observed option prices. In the context of this model calibration, the implied volatility enters the objective function. Instead of minimizing (for instance the quadratic) difference of model and market prices, the difference of the corresponding implied volatilities is used. This is a convenient normalization since options from deep in the money to far out of the money are transformed to the same scale. For calibration purposes, the implied volatility needs to be available rapidly---especially in view of routinely processed intraday recalibrations. Depending on the pricing routine employed, the accuracy needs to be medium or high.  Moreover, a closed-form of the implied volatility function is advantageous since it allows the implementation of gradient-based optimization routines. 

Unfortunately, the solution of this inverse problem is not available in an explicit form and thus a numerical approximation method is required. Since the implied volatility function is a crucial element of any financial toolbox, special care is called for. 
The method must allow the computation of implied volatilities for options in all of the different markets. Hence options  with very low or high volatilities as well as options with moneyness varying from far out of the money to deep in the money have to be included. Therefore the method must cover a \textit{large domain of input variables}. In order to satisfy the needs of the different applications the method should be \textit{highly efficient} for a given requirement in terms of accuracy. Even for very large data sets the method must be able to deliver accurate \textit{real-time evaluations} of the implied volatility. In view of the implied volatility as an ingredient of optimization routines, the approximation should be given in \textit{closed-form} with accessible derivatives. Finally, the method should be \textit{easy to implement and to maintain}. There exists a long list of papers dealing with this problem.

The first class of methods to determine the implied volatility are iterative root finders such as
\begin{itemize}
\item Newton-Raphson,
\item \textit{Matlab}s implied volatility function \textit{blsimpv},
\item the iterative methods of \cite{Jaeckel2006} and \cite{Jaeckel2015}.
\end{itemize}

The first approach dates back to \cite{ManasterKoehler1982} who showed that a Newton-Raphson algorithm can be applied to calculate the implied volatility. The \textit{blsimpv} function is part of the financial toolbox in \textit{Matlab} and uses an iterative scheme based on Brent-Dekker. The \textit{blsimpv} function becomes very slow for larger data-sets and the Newton-Raphson algorithm is highly dependent on the starting value of the iteration. For many standard parameters it often converges fast but for more extreme parameters, the number of iterative steps increases significantly, see Section \ref{section_numerik}.

To overcome this problem, \cite{Jaeckel2006} exploits the limit behaviour of the normalised call price to provide a better initial guess, which reduces the iterative steps in a modified Newton method. In \cite{Jaeckel2015} this approach is further improved using rational approximation for the initial guess and Householder's method for the iteration. This reduces the number of iterative steps even further. One drawback of the method is that it comes with the burden of a relatively complex implementation and therefore a costly maintenance. Already the generation of the initial guess relies on the rational cubic interpolation of \cite{DelbourgoGregory1985} and a transformation, which is highly sensitive in terms of the accuracy of the error function and the inverse of the normal distribution.
 
The second class of methods to compute the implied volatility are non-iterative approximations methods. These methods are popular since they provide
\begin{itemize}
\item fast computation of implied volatilities,
\item easy implementation and maintenance,
\item closed-form expressions,
\item a simple interpretation of the formula.
\end{itemize}

First, analytical approximations for at the money and later near the money options have been developed. Typically, these methods depend on a series expansion of the call price at the money. Prominent examples are the approximation formulas of \cite{BrennerSubrahmanyan1988}, \cite{Chance1996}, \cite{CorradoMiller1996}, \cite{ChambersNawalkha2001} and \cite{LorigPagliaraniPascucci2014}. Typically, these methods suffer from a poor performance for out of the money options. More recently, \cite{Li2008}, \cite{PistoriusStolte2012} and \cite{Celis2017} have developed rational approximations of the implied volatility. Unfortunately, the domain for which the latter set up the interpolation is very restrictive and excludes option prices which occur in practice.  In particular, options with relatively high or low volatilities cannot be handled. For example Figure \ref{fig:DAXDomains} illustrates the moneyness and implied time-scaled volatilities of options on the DAX index traded on 6/20/2017 (Source Thomson Reuters Eikon). In this example, only $85\%$ of all put options and $92\%$ of the call options are covered. Although the domain was designed for equity options, even in this case the formula cannot be applied to all relevant contracts.
Moreover, one needs additional iterative Newton steps to achieve a high accuracy close to machine precision for the methods of \cite{Li2008} and \cite{Celis2017}.
\begin{center}
\includegraphics[width=15cm]{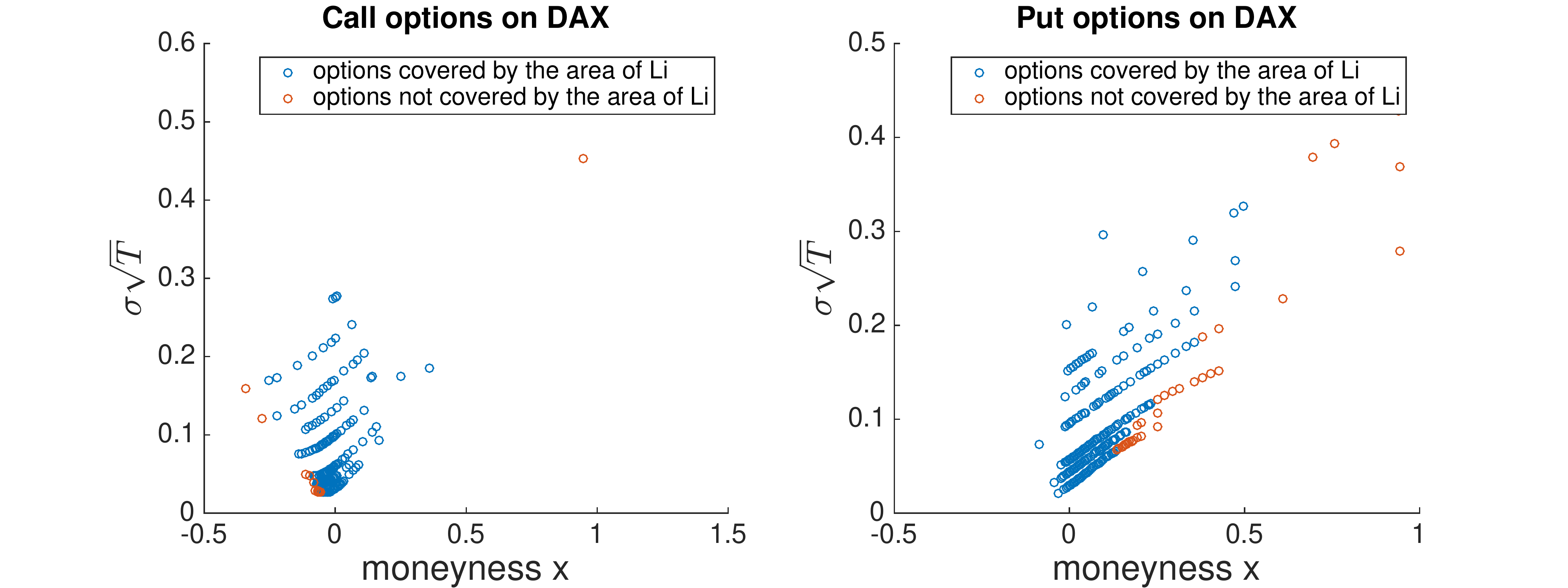}
\captionof{figure}{Moneyness $x$ and time-scaled volatility $\sigma\sqrt{T}$ of DAX options on 6/20/2017. We only considered options with positive trading volume.}\label{fig:DAXDomains}
\end{center}

In this paper, we propose polynomial approximation to the implied volatilities surface choosing Chebyshev interpolation. The approximation of the implied volatility thus inherits the appealing properties of Chebyshev interpolation, namely the fact that the approximation is highly efficient, stable and easy to implement. It is sufficient to invert the normalized call price, which reduces the dimensionality of the approximation to a bivariate Chebyshev interpolation. For this we use the algorithm provided in the \textit{MATLAB} package \textit{chebfun} (www.chebfun.org) that exploits the low-rank structure of the problem. Hence, the method enables a fast computation of implied volatilities at a high accuracy. In order to cover the whole range of relevant options, one has to investigate the shape of the call price surface further. We observe areas where the call price is almost linear as well as areas where the call price is extremely flat in the volatility. For an optimal treatment of the different areas we introduce a domain splitting. In the flat areas, we exploit the limit behaviour by introducing appropriate transformations. We show that the error of the interpolation decays subexponentially fast and we provide an explicit error bound. It is straightforward to adjust the method to any pre-set accuracy to obtain an optimal efficiency. Furthermore, the implied volatility function is represented by a polynomial and hence very easy to handle. Let us emphasize that this procedure is more general and can be applied to similar problems as well. To illustrate this, we approximate the implied volatility in a market model based on a Laplacian density function instead of a normal distribution introduced by \cite{Madan2016}.

The remainder of the article is as follows. In Section \ref{section_preliminaries}, we recall the normalized call price and the Chebyshev function on which our approach relies. In Section \ref{section_simple}, we introduce a simple, bivariate interpolation of the implied volatility based on a low-rank interpolation in Chebyshev nodes. We highlight the potential of the method and show that we reach a maximal error close to machine precision with a low number of interpolation points. 
In Section \ref{section_splitting}, we introduce the bivariate interpolation on a larger domain which includes very low and very high volatilities as well as deep in the money and far out of the money options. In Section \ref{section_numerik}, we show that the method is both, fast and accurate and compare it to the methods of Newton-Raphson, \cite{Li2008} and \cite{Jaeckel2015}. We devote the last section to the approximation of the implied volatility in the Laplacian market model.

\section{Preliminaries}\label{section_preliminaries}
\subsection{The normalized Black-Scholes price}\label{section_norm_bs_price}
As stated, the implied volatility depends on the parameters $S_{0}$,$K$, $T$, $r$ and the option premium $C$. The computational effort to interpolate a function depending on five variables is challenging. Fortunately, we can reduce the dimensionality as stated in \cite{Jaeckel2015} amongst others using the normalized call price given as
\begin{eqnarray}\label{Normalized BS-formula}
c(x,v)&=&e^{\frac{x}{2}}\Phi\left(\frac{x}{v}+\frac{v}{2}\right)-e^{-\frac{x}{2}}\Phi\left(\frac{x}{v}-\frac{v}{2}\right)\text{ with}\\
\nonumber \\
x &=& \log(S_0e^{rT}/K)=rT+\log(S_0/K)\nonumber\\
v&=&\sigma\sqrt{T}\nonumber.
\end{eqnarray}
In this context $x$ measures the \textit{moneyness} (the option is out of the money if $x<0$, at the money if $x\approx0$ and in the money if $x>0$), $v$ corresponds to the \textit{time-scaled volatility}. We have
\begin{align}\label{Normalization}
c(x,v)
=\frac{C(S_0,K,T,r,\sigma)}{\sqrt{S_0 e^{-rT}K}}
\end{align}
Furthermore, call prices of in the money options can be expressed by those of out of the money options, namely
\begin{align}\label{ITM_option_formula}
c(-x,v)=c(x,v)+e^{-\frac{x}{2}}-e^{\frac{x}{2}}.
\end{align}
Hence the domain can be reduced to $x\leq0$ and consequently the call price is normalized to values in $[0,1]$. To calculate the implied volatility $\sigma$ for a call price $C$ it is thus sufficient to solve Equation \eqref{Normalized BS-formula} for $v$ using the normalized call price $c$.

\subsection{Chebyshev Interpolation}
The polynomial interpolation of a function $f$ on $[-1,1]$ in the \textit{Chebyshev points} $x_{k}=\cos(k\pi/N)$ is given by
\begin{align}\label{1D_Cheby_interp}
f(x)\approx I_N(x):=\sum_{j=0}^N a_j T_j(x) \qquad \text{with} \qquad  a_j= \frac{2^{\1{_{0<j<N}}}}{N}\sum_{k=0}^N { }^{\prime\prime} f(x_k)T_{j}(x_k),
\end{align}
where $T_{j}(x)=\cos(j\cos^{-1}2(x))$ and $\sum { }^{\prime\prime} $ indicates that the first and the last summand are halved. If the function has an analytic extension to a Bernstein ellipse $E_\rho$, the error decays exponentially, see Theorem 8.2 of \cite{Trefethen2013}. In practice, this often yields an approximation close to machine precision with a low interpolation order. Together with a stable implementation being available, see \cite{Higham2004}, these are the key advantages of the Chebyshev interpolation that we will exploit.

The univariate Chebyshev interpolation admits a two-dimensional tensor based extension. A function $f:[-1,1]^2\to\IR$ can be approximated by the interpolation
\begin{align}\label{eq:2D_Cheby_tensor}
f(x,y)\approx I^{N_1,N_2}(x,y):=\sum_{i=0}^{N_1-1}\sum_{j=0}^{N_2-1} a_{ij}T_i(x)T_j(y).
\end{align}
with two-dimensional coefficients given by
\begin{align*}
a_{ij}=\frac{2^{\1{_{0<i<N_{1}}}}}{N_{1}}\frac{2^{\1{_{0<j<N_{2}}}}}{N_{2}}\sum_{k_{1}=0}^{N_{1}} { }^{\prime\prime}\sum_{k_{2}=0}^{N_{2}} { }^{\prime\prime} f(x_{k_1},y_{k_2})T_{i}(x_{k_1})T_{j}(y_{k_2}).
\end{align*}

Again, we obtain an subexponential error decay if the function has an analytic extension to a two-dimensional Bernstein ellipse, see \cite{SauterSchwab2010}. The tensor approach of \eqref{eq:2D_Cheby_tensor} suffers from the curse of dimension: To decrease the error in the same proportion as in the univariate case, the number of summands and thus the complexity increases quadratically. Therefore more efficient bivariate Chebyshev interpolations have been developed. In particular, the algorithm of \cite{TownsendTrefethen2013} implemented in \textit{chebfun2} reconciles the opposed aims of high accuracy and high efficiency for bivariate functions. It relies on a Gauss elimination with complete pivoting to find an optimal low rank $k$ approximation. This leads to
\[
f(x,y)\approx f_k(x,y):=\sum_{j=1}^kd_j c_j(y) r_j(x)
\]
where $c_j$ and $r_j$ are one-dimensional Chebyshev interpolations of degree $N_{1}$ and $N_{2}$. This enables a matrix representation of the resulting interpolation.

\section{Introduction of the approximation method}\label{section_simple}

We introduce a direct interpolation of the implied volatility function using Chebyshev nodes. The two-dimensional Chebyshev interpolation requires the function to be defined on the rectangle $[-1,1]\times[-1,1]$. For the implied volatility $v(x,c)$ this is not given a priori. The variable $x$ can easily be restricted to some interval $x\in[x_{min},x_{max}]$ which can be transformed to $[-1,1]$ by a linear transformation $\varphi$,
\begin{align}\label{eq:trans_phi_x}
\varphi:[x_{min},x_{max}]\to[-1,1] \qquad \text{with} \qquad \varphi(x):=1-2\cdot\frac{x_{max}-x}{x_{max}-x_{min}}.
\end{align}
The maximal domain of $c$, on the contrary, does depend on $x$ as for $x<0$ the upper limit is given by $e^{\frac{x}{2}}$.


The intuitive approach  is to choose $\xi \in[\xi_{min},\xi_{max}]$ with $c=\xi e^{\frac{x}{2}}$ for a given moneyness $x\in[x_{min},x_{max}]$ and scale the resulting interval $[\xi_{min}e^{\frac{x}{2}},\xi_{max}e^{\frac{x}{2}}]$ to $[-1,1]$ by a linear transformation. If $\xi_{min}$ is not chosen to close to 0, a two-dimensional Chebyshev interpolation on this domain provides promising results.\\

For a first numerical example, we fix $x_{min}=-5$, $x_{max}=0$, $\xi_{min}=0.05$ and $\xi_{max}=0.8$. Then we choose a $50\times50$ Chebyshev grid $(\tilde{x}_{ij},\tilde{c}_{ij})\in[-1,1]^2$ and transform the points to the domain by setting $x_{ij}:=x_{min}+\frac{1}{2}(\tilde{x}_{ij}+1)(x_{max}-x_{min})$ and $c_{ij}:=\xi_{min}e^{\frac{x_{ij}}{2}}+\frac{1}{2}(\tilde{c}_{ij}+1)(\xi_{max}e^{\frac{x_{ij}}{2}}-\xi_{min}e^{\frac{x_{ij}}{2}})$. On these points we compute the implied volatilities using the method of \cite{Jaeckel2015} and apply the \textit{chebfun2}-algorithm.

To determine the interpolation error we define an equidistant grid of $100$ points in the interval $[x_{min},x_{max}]$. For fixed $x$, the interval bounds in $v$ are defined as $v_{min}(x)=v\left(\xi_{min}e^{\frac{x}{2}},x\right)$ and $v_{max}(x)=v\left(\xi_{max}e^{\frac{x}{2}},x\right)$. For each $x$-value in the fixed equidistant grid, $100$ points distributed equidistantly in $[v_{min}(x),v_{max}(x)]$ are determined. This leads to $100\times100$ points in the $(x,v)$ space as reference points for which we compute normalized call prices $c(x,v)$. For each reference call price we compute the implied volatility using the bivariate Chebyshev method.\\
Figure \ref{fig:IntuitiveApproachSmallArea} shows that this approach performs very well. The maximal error lies below a level of $10^{-7}$ for $N=N_1=N_2=50$ and decreases exponentially fast in $N$.

%

\begin{figure}[htb]
    \centering
    \begin{minipage}{0.45\linewidth}
        \centering
        \includegraphics[width=7cm]{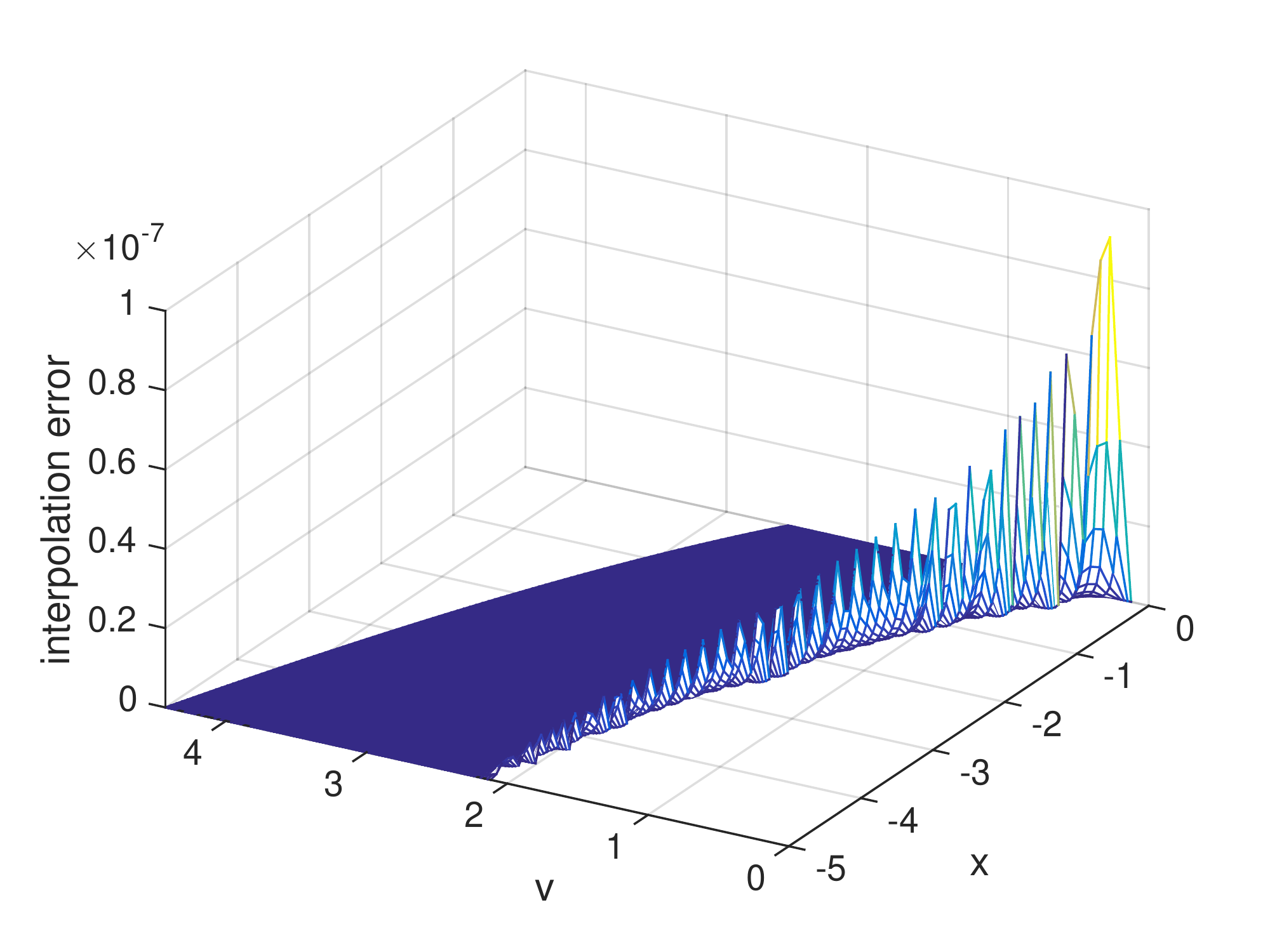}
    \end{minipage}
    \begin{minipage}{0.45\linewidth}
        \centering
        \includegraphics[width=7cm]{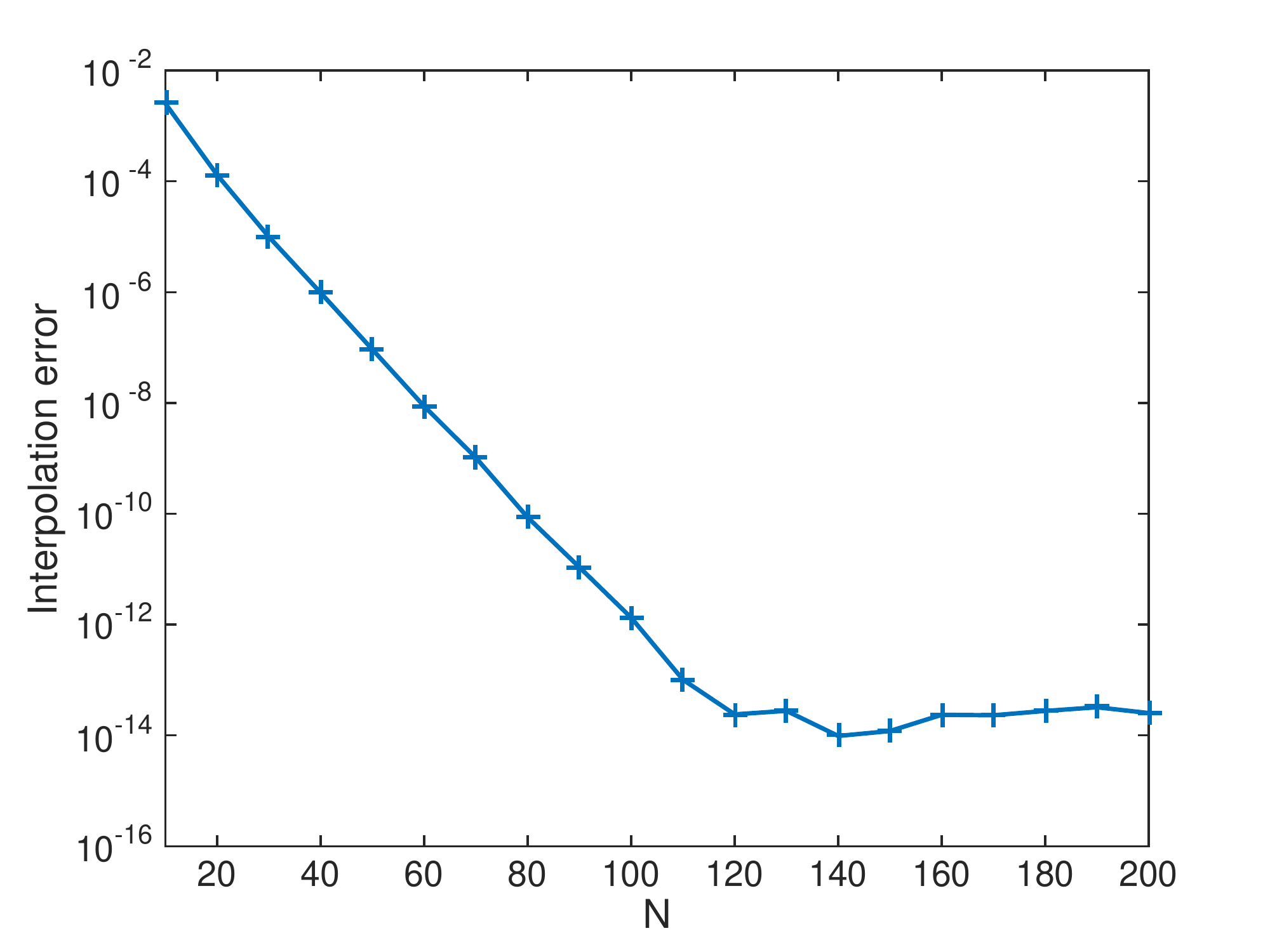}
    \end{minipage}
\captionof{figure}{Interpolation error (left) and exponential error decay (right) using linear transformations for $x\in[x_{min},x_{max}]$, $c\in[\xi_{min}e^{\frac{x}{2}},\xi_{max}e^{\frac{x}{2}}]$ with $x_{min}=-5$, $x_{max}=0$, $\xi_{min}=0.05$ and $\xi_{max}=0.8$.}\label{fig:IntuitiveApproachSmallArea}%
\end{figure}

As in the approximation methods mentioned above, we have pre-fixed a domain that is convenient for the approach. Naturally, the question arises as to which domain is appropriate to cover the relevant option data. 

\subsection{Investigation of the interpolation domain by market data}\label{section_data_domain}
To find an appropriate interpolation domain, we investigate option data of the DAX, the EURO STOXX 50, the S\&P 500 and the VIX index from \textit{Thomson Reuters Eikon}. For all options with non-zero trading volume we compute the forward moneyness $x$ and the time-scaled volatility $\sigma\sqrt{T}$. Then we check if the resulting parameters are covered by the domain of Li. Figure \ref{fig:Domains_Data} illustrates the option parameters for all four indices. For all indices we observe that a relevant part of the options is not covered by the domain of Li. We observe moneyness between $-1.5$ and $2$ as well as time-scaled volatilities up to $1$. In different markets or under different market conditions one can expect to observe even more extreme option parameters. Volatilities become considerably higher during a financial crisis. This motivates us to set up a Chebyshev interpolation of the implied volatility on a significantly larger domain which covers all relevant option data. To do this in the most efficient way we need to enhance the intuitive approach introduced above with a splitting of the domain and tailored scaling functions. 

\begin{figure}[h]
\centering
    \begin{minipage}{0.49\linewidth}
        \centering
        \includegraphics[width=8.8cm]{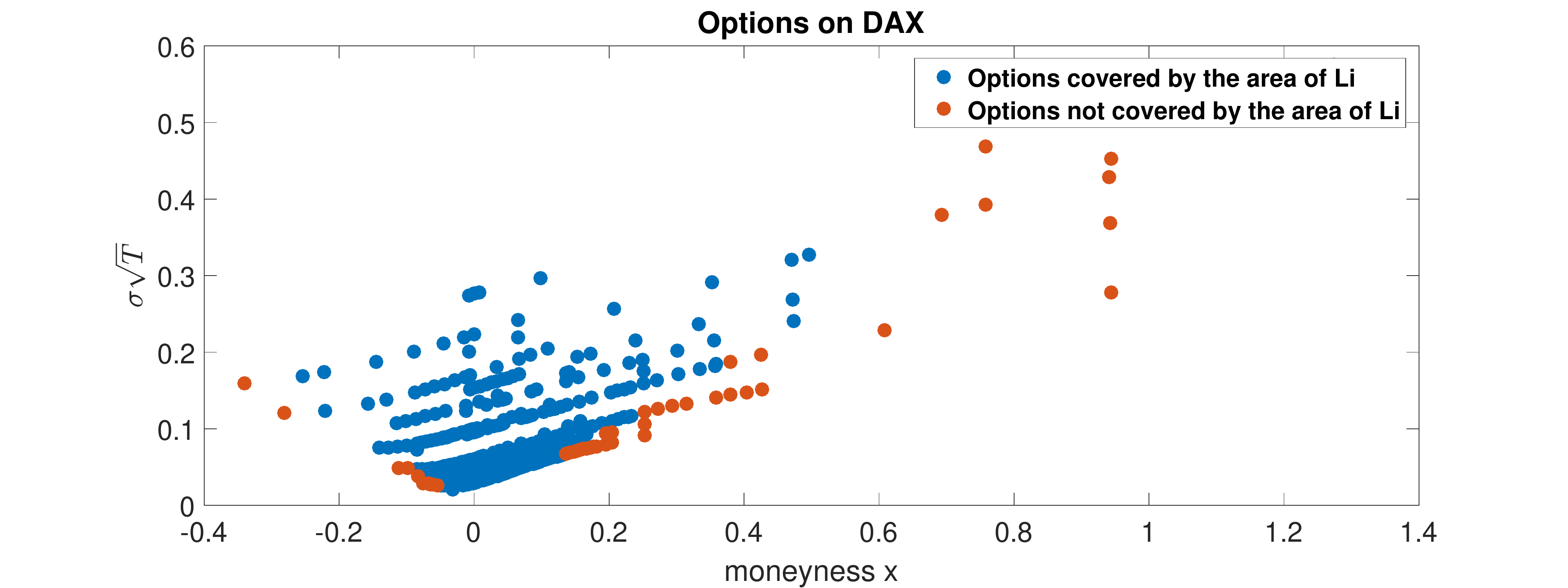}
    \end{minipage}
    \hfill
    \begin{minipage}{0.49\linewidth}
        \centering
        \includegraphics[width=8.8cm]{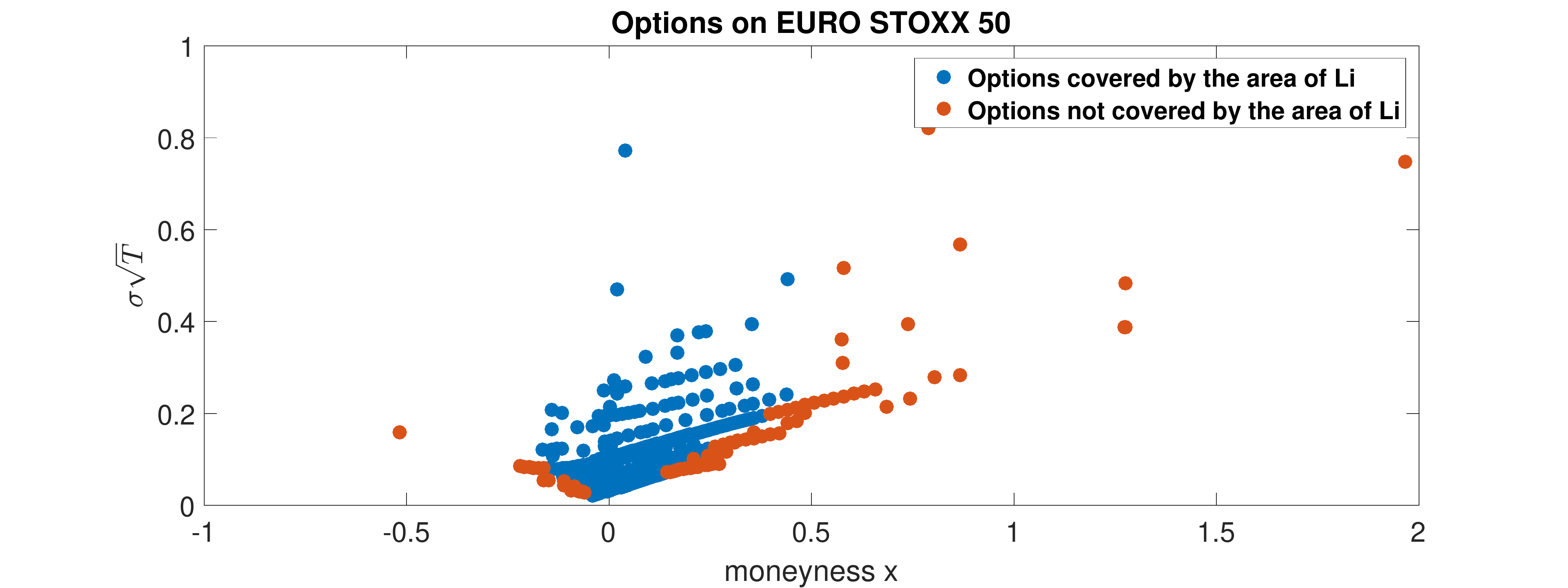}
    \end{minipage}
    
    \begin{minipage}{0.49\linewidth}
        \centering
        \includegraphics[width=8.8cm]{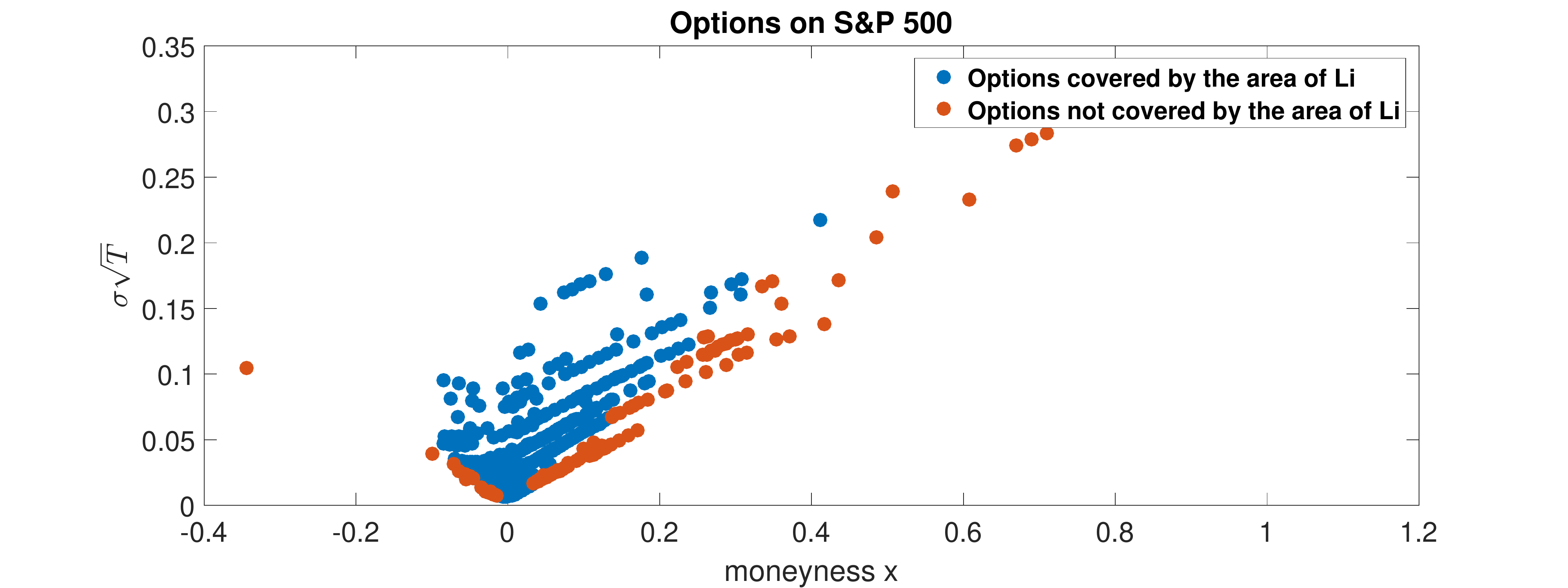}
    \end{minipage}
    \hfill
    \begin{minipage}{0.49\linewidth}
        \centering
        \includegraphics[width=8.8cm]{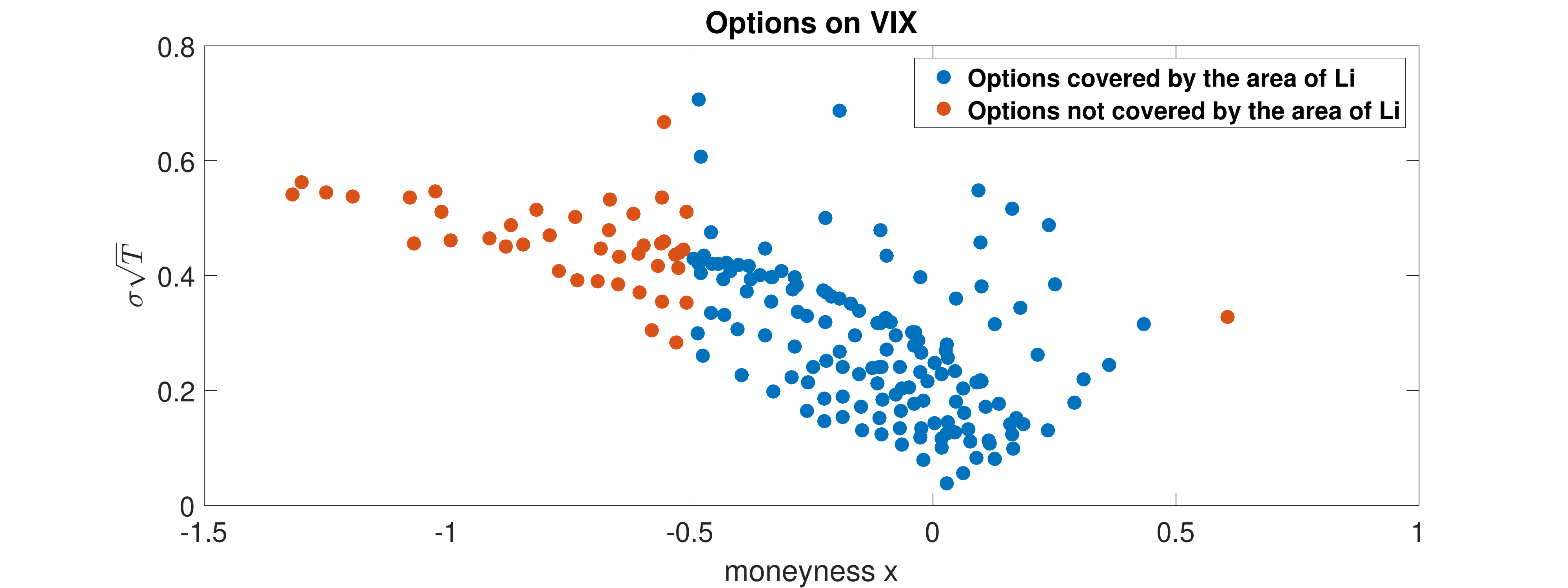}
    \end{minipage}
\captionof{figure}{Moneyness $x$ and time-scaled volatility $\sigma\sqrt{T}$ of options on four different indices. We only considered options with positive trading volume.}\label{fig:Domains_Data}%
\end{figure}

\section{Domain splitting and scaling}\label{section_splitting}
To derive an approximation of the implied volatility on a sufficiently large domain, we further inspect the normalized call price. The implied volatility is not analytic at $c(x)=0$ and $c(x)=e^{\frac{x}{2}}$. Therefore the maximal possible interval needs to be restricted to $0<v_{min}(x)<v_{max}(x)<\infty$ with call prices $0<c_{min}<c_{max}<e^{\frac{x}{2}}$, which excludes these points. This assumption is not restrictive if the chosen $v_{min}$ is small enough. Extending the domain towards the maximal interval decreases the rate of convergence. To reduce this impact, we exploit the limit behaviour of the call price. Graph \ref{fig:CallPriceSplit} shows for a fixed moneyness $x$ the normalized call price as a function of the volatility. We observe that the call price is flat for very low as well as very high volatilities and almost linear around the point of inflection. This motivates us to split the domain into three parts. \\

\begin{center}
\includegraphics[width=15cm]{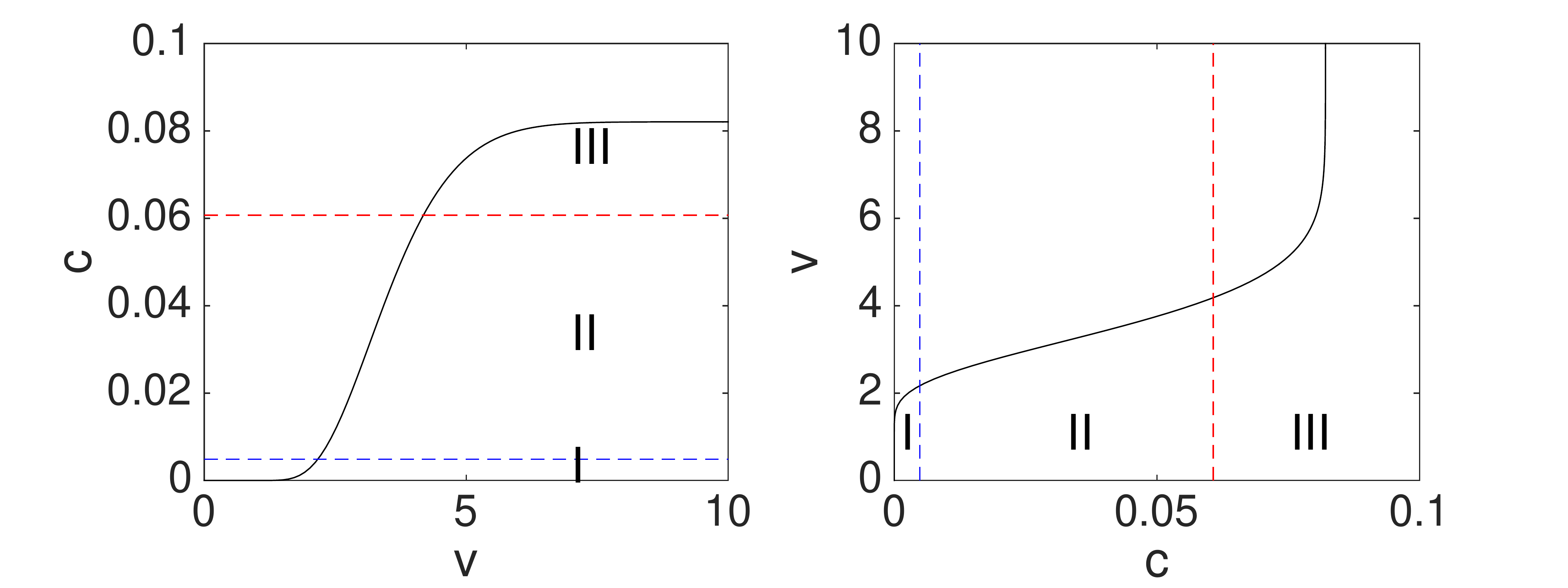}
\captionof{figure}{Splitting of the normalized call price ($c$) depending on the time-scaled volatility ($v$) and its inversion for $x=-5$ into three parts.}\label{fig:CallPriceSplit}
\end{center}

\begin{align}
D_1:=[c_{min}(x), c_1(x)],\quad D_2:=[c_1(x), c_2(x)],\quad D_3:=[c_2(x), c_{max}(x)]
\end{align}
with corresponding volatilities $0<v_{min}(x)<v_1(x)<v_2(x)<v_{max}(x)<1$. The idea of splitting the domain is based on the method of \cite{Jaeckel2015}.

For each domain we will tailor a bivariate Chebyshev interpolation. Where call prices are flat its inverse becomes very steep. Hence, a direct polynomial interpolation is not well-suited. Fortunately, by exploiting the asymptotic behaviour of the call price function, we resolve the problem. On each interval, we define a scaling function $\phi_{i,x}:D_i \to[-1,1]$ for $i\in\{1,2,3\}$ which transforms the call price to $[-1,1]$ for each $x\in[x_{min},x_{max}]$. For the resulting functions $\tilde{v}:[-1,1]^2\to \IR$, $(\tilde{c},\tilde{x})\mapsto v(c,x)$ with $x=\varphi^{-1}(\tilde{x})$ and $c=\phi_{i,x}^{-1}(\tilde{c})$ for $i\in\{1,2,3\}$ where $\varphi$ is the linear scaling of \eqref{eq:trans_phi_x}. For a given call price $c$ and moneyness $x\leq0$ the implied volatility can then be approximated by 
\[
v(c,x)\approx I_i^{N^i_1,N^i_2}(\phi_{i,x}(c),\varphi(x))\text{ where $i$ satisfies }c\in D_i.
\]

\subsection{Scaling functions}
In the following, we introduce the appropriate scaling functions for each of the areas.

\subsubsection{Medium volatilities}
First consider the middle part of the function. As discussed, for $v$ around the point of inflection, the implied volatility surface is almost linear. Thus, a linear scaling suffices,
\begin{align*}
\phi_{2,x}:[c_{1}(x),c_{2}(x)]\to[-1,1],\quad 
 c\mapsto 2\frac{c-c_{1}(x)}{c_{2}(x)-c_{1}(x)}-1.
\end{align*}
Clearly, $\phi_2$ is analytic and the inverse is given by
\begin{align*}
\phi_{2,x}^{-1}[-1,1]\to[c_{1}(x),c_{2}(x)],\quad 
 \tilde{c}\mapsto c_{1}(x)+\frac{1}{2}(\tilde{c}+1)(c_{2}(x)-c_{1}(x)).
\end{align*}

\subsubsection{Low volatilities}
For low volatilities the call price function is very flat, and thus the implied volatility function as its inverse is steep. Therefore, a linear scaling will not provide an appropriate transformation prior to a polynomial interpolation. Instead, we propose a suitable scaling function that reduces the steepness of the inverse such that it becomes almost linear. This will increase the efficiency of the resulting approximation considerably, when compared to a linear scaling. To do this, we explore the limit behaviour of the normalized call price. For $v\to0$ we have by equation (2.8) of \cite{Jaeckel2006} that
\begin{align*}
c(x,v)\approx \varphi\left(\frac{x}{v}\right) \left(\frac{v^3}{x^2}\right),
\end{align*}
where $\varphi$ is the density of the standard normal distribution. By inverting the function $c(v)=\varphi\left(\frac{x}{v}\right)$, which has the major effect in the limit, we obtain an inverse of the form $v=c_{2}(-(c_{1} + 2*log(c)/x^2))^{-1/2}$ with constants $c_{1},c_{2}$ that are not relevant for us. This leads to the following transformation
\begin{align*}
\tilde{\phi}_{1,x}:[0,c_{1}(x)]&\rightarrow[-1,1]\\
 c&\mapsto\begin{cases}2\left(-\frac{2}{(x-\delta)^2}\log(c)+\frac{2}{(x-\delta)^2}\log(c_{1}(x))+1\right)^{-\frac{1}{2}}-1 &\text{ if }c>0\\
 -1 &\text{ else.}
 \end{cases}
\end{align*}
The parameter $\delta>0$ ensures the well-definedness for $x=0$ and the remaining terms are needed to map the interval $[0,c_{1}(x)]$ to $[-1,1]$. The transformation $\tilde{\phi}_{1,x}$ is analytic with inverse
\begin{eqnarray*}
\tilde{\phi}_{1,x}^{-1}:[-1,1]&\to[0,c_{1}(x)]:\ 
 \tilde{c}&\mapsto \begin{cases} c_{1}(x)e^{-\frac{2(x-\delta)^2}{(\tilde{c}+1)^2}+\frac{(x-\delta)^2}{2}} &\text{ if } \tilde{c}>-1\\
 0 &\text{ else.}
 \end{cases}
\end{eqnarray*}
Using this transformation the function $v(\tilde{\phi}_{1,x}^{-1}(\tilde{c}),x)$ is approximately linear in $\tilde{c}$.  

As already mentioned, to guarantee analyticity we restrict the interval to $[c_{min}(x),c_{1}(x)]$ for $0<c_{min}(x)<c_1(x)<e^{\frac{x}{2}}$. Therefore, we define the scaling function for the low volatilities $\phi_{1,x}:[c_{min}(x),c_{1}(x)]\to[-1,1]$ as $\phi_{1,x}(c):= l(\tilde{\phi}_{1,x}(c))$, where $l$ is the linear transformation
\begin{eqnarray*}
	l:[\tilde{\phi}_{1,x}(c_{min}(x)),1]\to[-1,1]:
	c\mapsto 2\cdot\frac{c-\tilde{\phi}_{1,x}(c_{min}(x))}{1-\tilde{\phi}_{1,x}(c_{min}(x))}-1.
\end{eqnarray*}
The function $c\mapsto\phi_{1,x}(c)$ is analytic in the interval $[c_{min}(x),c_{1}(x)]$ as it is a composition of two analytic functions. The inverse of $\phi_{1,x}$ is given by $\phi_{1,x}^{-1}(\tilde{c})=\tilde{\phi}_{1,x}^{-1}\left(l^{-1}\left(\tilde{c}\right)\right)$. 

\subsubsection{High volatilities}
Just as for the low volatilities, the call price function is very flat for high volatilities and thus its inverse becomes steep. As $\lim_{c\to e^{\frac{x}{2}}}v(c,x)=\infty$ the implied volatility function is not even bounded. As a first step, the volatility is capped by some $v_{max}$ to guarantee that the slope will not be arbitrarily high. Again, a linear transformation is not the best choice and we propose a different scaling based on the behaviour of the call price. From \cite{Jaeckel2006} equation (2.7) we obtain for $v\to\infty$ 
\[
c(x,v)\approx e^{\frac{x}{2}}-\frac{4}{v}\varphi\left(\frac{v}{2}\right).
\]
A similar transformation as in the case of low volatilities entails improvement. Assume first that $c_{max}(x)=e^{\frac{x}{2}}$ and define
\begin{eqnarray*}
\tilde{\phi}_{3,x}:[c_2(x),e^{\frac{x}{2}}]\to[0,\infty]:
 c\mapsto \begin{cases} \left(-8\log\left(\frac{e^{\frac{x}{2}}-c}{e^{\frac{x}{2}}-c_2(x)}\right)\right)^{\frac{1}{2}} & \text{ if }c<e^{\frac{x}{2}}\\
 \infty & \text{ else}
 \end{cases}
\end{eqnarray*}
with inverse
\begin{eqnarray*}
\tilde{\phi}_{3,x}^{-1}:[0,\infty]\to[c_2(x),e^{\frac{x}{2}}]:
 \tilde{c}\mapsto \begin{cases} e^{\frac{x}{2}}-\left(e^{x/2}-c_2(x)\right)e^{-\frac{\tilde{c}^2}{8}} & \text{ if }\tilde{c}<\infty \\
 e^{\frac{x}{2}} & \text{ else.} \end{cases}
\end{eqnarray*}
Exploiting the limit behaviour of the call price, one can show that for $v$ large enough $\tilde{c}=\tilde{\phi}_{3,x}(c(x,v))\approx -v$. Hence $v=v(\tilde{\phi}_{3,x}^{-1}(\tilde{c}),x)\approx -\tilde{c}$ which is linear in $\tilde{c}$.

Now for $c_{max}(x)<e^{\frac{x}{2}}$ the transformation is a bijection into a bounded domain which can be normalized to $[-1,1]$  by the linear transform as in the previous case
\begin{eqnarray*}
	l:[0,\tilde{\phi}_{3,x}(c_{max}(x))]\to[-1,1]:
	c\mapsto\frac{2c}{\tilde{\phi}_{3,x}(c_{max}(x))}-1.
\end{eqnarray*}
Thus $\phi_{3,x}(c):= l(\tilde{\phi}_{3,x}(x))$ and $\phi_{3,x}^{-1}(\tilde{c})=\tilde{\phi}_{3,x}^{-1}\left(l^{-1}\left(\tilde{c}\right)\right)$ depending on the choice of $c_{max}$.


\subsection{Splitting} \label{ch:Splitting}
The explicit choice of the boundaries depends on the particular application. In the following we want to set the boundaries in such a way that a very large set of parameters is covered and the rate of convergence is about the same for all areas.\\ 

\textbf{Maximal volatility $v_{max}$:}\\
We choose as an upper bound for the time scaled volatility $v_{max}=6$. This allows us to include highly volatile markets and long maturities. At the same time the method can achieve accuracies close to machine precision.\\

\textbf{Minimal volatility $v_{min}$:}\\
We define a lower bound by
\[
v_{min}(x)=0.001-0.03x.
\]
For this choice the corresponding prices $c_{min}(x)$ can be computed with the standard machine precision. It includes very low volatilities. For instance at $x=\log(e^{rT}S_{0}/K)=0$ this choice even allows call options with a time to maturity of one day ($T=1/365$) and a Black-Scholes volatility of $\sigma\approx2\%$. The rate of convergence can be increased further if $v_{min}$ is chosen higher.\\
%

\textbf{Splitting volatilities $v_1$ and $v_2$:}\\
We choose $v_1$ and $v_2$ according to the properties of the call price function. The call price function has a unique inflection point for $v_c(x)=\sqrt{2\vert x\vert}$ where the slope is maximal. \cite{Jaeckel2015} proposes the lower bound $v_1$ as the zeros of the tangent line at this point. The upper bound $v_2$ is set to be the point where the line hits the maximal call price depending on $x$. See Figure \ref{fig:SplittingTangentLine}. The tangent line is given as
\[
f(v)=\frac{\partial}{\partial v}c(x,v_c)\left( v-v_c(x)\right) +c(x,v_c).
\]
Thus 
\begin{align*}
	\tilde{v}_{1}(x)=v_c(x)-\frac{c(x,v_c)}{\frac{\partial}{\partial v}c(x,v_c)} \qquad
	\tilde{v}_{2}(x)=v_c(x)+\frac{e^{\frac{x}{2}}-c(x,v_c)}{\frac{\partial}{\partial v}c(x,v_c)}
\end{align*}

\begin{figure}[htb]
    \centering
    \includegraphics[width=12cm]{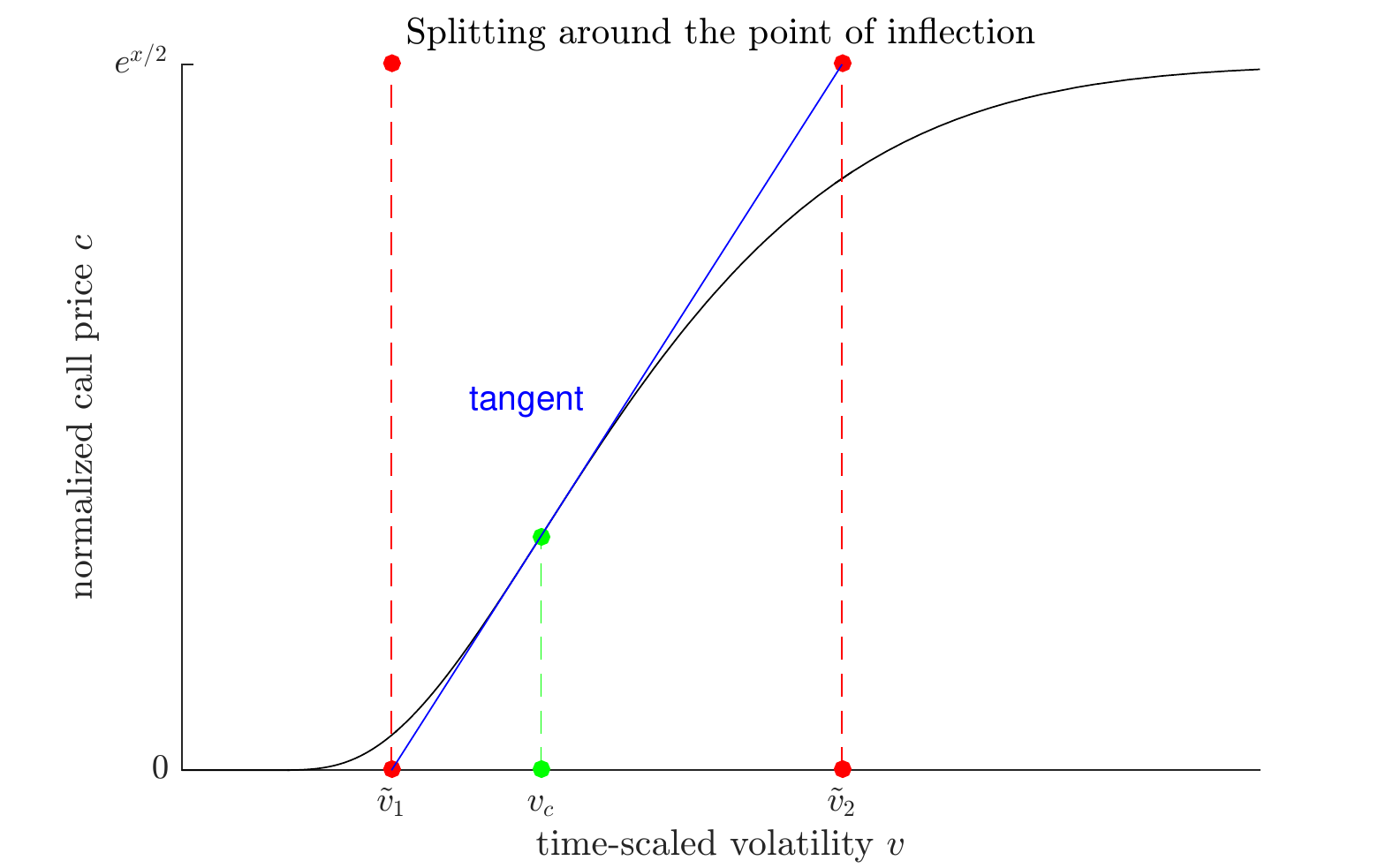}
\captionof{figure}{Definition of the splitting at $\tilde{v}_{1}$ and $\tilde{v}_{2}$ by the zeros of the tangent line at the point of inflection $v_c$.}\label{fig:SplittingTangentLine}%
\end{figure}

However, this choice of boundaries has two serious disadvantages. First, the boundary $\tilde{v}_{1}$ tends to zero, hence for small values of $x$ we obtain $\tilde{v}_{1}(x)<v_{min}(x)$. Second, the computation of $\tilde{v}_{1}(x)$ and $\tilde{v}_{2}(x)$ requires the evaluation of $c(x,v_c)$ and $\frac{\partial}{\partial v}c(x,v_c)$ for each $x$. For real-time computation on large data sets, this becomes a computational burden. We solve this problem by replacing $\tilde{v}_{1}$ and $\tilde{v}_{2}$ with linear approximations. We propose the boundaries
\begin{align*}
v_1(x)=0.25-0.4x \qquad \text{and} \qquad v_{2}(x)=2-0.4x.
\end{align*}

\textbf{Splitting of the low volatility area}\\
For low volatilities we improve the interpolation by introducing a further splitting in $x$. The behaviour of the function changes at the point of inflection. As shown before, we need to set $v_1(0)>v_{min}(0)>0$. Thus at some point the interpolation boundary $v_1$ will cross this change of behaviour. This can be anticipated by a splitting at the point $x$ where $v_{c}(x)=v_1(x)$. For the proposed linear splitting these points are given by $-11.2152$ and $-0.0348$. The first point is outside the domain $[-5,0]$ for $x$, hence we only consider the second point. We divide the area of the low volatilities in an Area I for $x\in[-5,-0.0348]$ and an Area I' for $x\in[-0.0348,0]$, see Figure \ref{fig:PlotAreas}. The empirical results show that this additional splitting further improves the rate of convergence.


\section{Error analysis} 
The following theorem is the theoretical foundation of the high efficiency of the approximation method. Thanks to the analyticity of the Black-Scholes call price and the scaling functions, we gather that the convergence is sub-exponential in the number of nodal points.
\begin{theorem}\label{th:volasErrorDecay}
Let $\phi_i^{-1}(\tilde{c},\tilde{x})$ be analytically continuable to some open region around $[-1,1]^2$ and let $0<\phi_i^{-1}([-1,1],x)<e^{\frac{x}{2}}$ for each $x\in[-1,1]$. Then there exist constants $\rho_1,\rho_2>1, V>0$ such that for $\tilde{v}(\tilde{c},\tilde{x}):=v(\phi_i^{-1}(\tilde{c},\tilde{x}),\phi_x^{-1}(\tilde{x}))$ and its bivariate Chebshev interpolation $I_i^{N^i_1,N^i_2}(\tilde{c},\tilde{x}):=\sum_{j=0}^{N^i_1-1}\sum_{k=0}^{N^i_2-1} a_{jk}T_j(\tilde{c})T_k(\tilde{x})$
\[
\max_{(\tilde{c},\tilde{x})\in[-1,1]^2}|\tilde{v}(\tilde{c},\tilde{x})-I_i^{N^i_1,N^i_2}(\tilde{c},\tilde{x})|\leq4V\left( \frac{\rho_1^{-2(N_1-1)}+\rho_2^{-2(N_2-1)}}{(1-\rho_1^{-2})(1-\rho_2^{-2})}\right)^{\frac{1}{2}}.
\]

\end{theorem}
\begin{proof}
According to Lemma 7.3.3 of \cite{SauterSchwab2010} we need to show that $\tilde{v}(\tilde{c},\tilde{x}):=v(\phi_i^{-1}(\tilde{c}),\phi_x^{-1}(\tilde{x}))$ is analytically continuable and bounded on $E_{\rho_1}\times E_{\rho_2}$ where $E_{\rho_1}$ and $E_{\rho_2}$ are Bernstein ellipses. \cite{GassGlauMahlstedtMair2015} show that the call price is analytic. For fixed $\tilde{x}$ the implied volatility function $\tilde{v}$ is holomorphic in $\tilde{c}\in[-1,1]$ since the inverse of a bijective holomorphic function is again holomorphic. Next we need to prove analyticity in $\tilde{x}\in[-1,1]$. Let $\tilde{c}\in[-1,1]$. Define $F(x,v):=c(x,v)-\phi_i(\tilde{c},x)$. Then the function $v(\phi_i(\tilde{c},x),x)$ is implicitly given by the solution of $F(x,v)=0$. Furthermore, for each $x\in[x_{min},x_{max}]$, $F$ is holomorphic in some open region with
\[
\left\vert\frac{\partial}{\partial v} F(x,v)\right\vert=\left\vert\frac{\partial}{\partial v} c(x,v)\right\vert=\left\vert\frac{1}{\sqrt{2\pi}}e^{-\frac{x^2}{2v^2}-\frac{v^2}{8}}\right\vert>0
\]
as $v>0$. Thus by the complex implicit function theorem (see Theorem 7.6 of \cite{FritzscheGrauert2012}) there exits a unique function $v(\phi_i(\tilde{c},x),x)$ that is holomorphic in some region around $x$. Thus $\tilde{v}$ is holomorphic in  $G_1\times G_2$ where $G_1$ and $G_2$ are open regions of $[-1,1]$. Thus there exist $\rho_1,\rho_2>1$ such that $E_{\rho_1}\subset G_1$ and $E_{\rho_2}\subset G_2$. The boundedness follows for sufficiently small $\rho_1$, $\rho_2$ as $\tilde{v}$ is continuous on $[-1,1]^2$. 
\end{proof}
%
%
%
We can enhance the efficiency even further by exploiting the low-rank structure of the bivariate functions. To do so, in our implementation we use the \textit{chebfun2} algorithm based on \cite{TownsendTrefethen2013}.

\section{Implementation}\label{ch:Implementation}
As a starting point for the approximation of the implied volatility function, we split the interpolation domain into four different areas. For each area, we approximate the implied volatility by a separate bivariate Chebyshev interpolation of the form $v\approx I_i^{N^i_1,N^i_2}(\phi_{i,x}(c),\phi_x(x))$ where $\phi_x$ is defined as in \eqref{eq:trans_phi_x} and for each area we have a different scaling $\phi_{i,x}$ in $c$. For the sake of a lucid presentation, we list the different areas and transformations below.\\

\textbf{Area I:} For $x\in[-5,-0.0348]$ and $c\in[c_{min}(x),c_1(x)]$ we have
\begin{align*}
\phi_{1,x}(c):=2\cdot\frac{\tilde{\phi}_1(c)-\tilde{\phi}_1(c_{min}(x))}{1-\tilde{\phi}_1(c_{min}(x))}-1.
\end{align*}

\textbf{Area I':} For $x\in[-0.0348,0]$ and $c\in[c_{min}(x),c_1(x)]$ we again use transformation $\phi_{1,x}(c)$.\\

\textbf{Area II:} For $x\in[-5,0]$ and $c\in[c_1(x),c_2(x)]$ we have
\begin{align*}
\phi_{2,x}(c):=2\frac{c-c_{1}(x)}{c_2(x)-c_1(x)}-1.
\end{align*}

\textbf{Area III:} For $x\in[-5,0]$ and $c\in[c_1(x),c_{max}(x)]$ we have
\begin{align*}
\phi_{3,x}(c):=\frac{2\tilde{\phi_3}(c)}{\tilde{\phi}_3(c_{max}(x))}-1.
\end{align*}
The call prices $c_{min}(x),c_{1}(x),c_{2}(x)$ and $c_{max}(x)$ correspond to the volatilities
\begin{align*}
v_{min}(x)=0.001-0.03x, \quad v_{1}(x)=0.25-0.4x, \quad v_{2}(x)=2-0.4x, \quad v_{max}(x)=6.
\end{align*}
\begin{figure}[htb]
    \centering
    \begin{minipage}{0.45\linewidth}
        \centering
        \includegraphics[width=8.62cm]{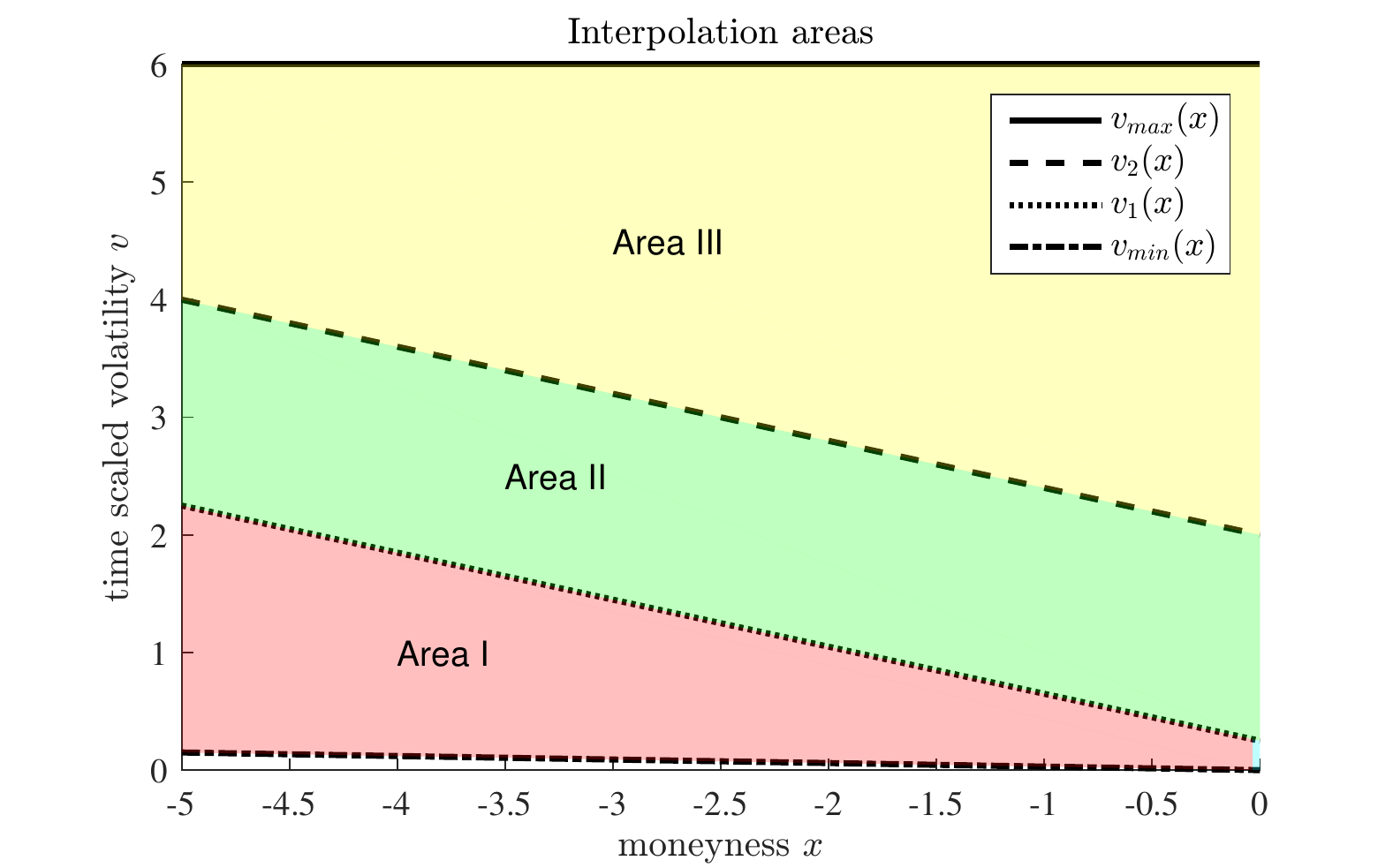}
    \end{minipage}
    \begin{minipage}{0.45\linewidth}
        \centering
        \includegraphics[width=5.38cm]{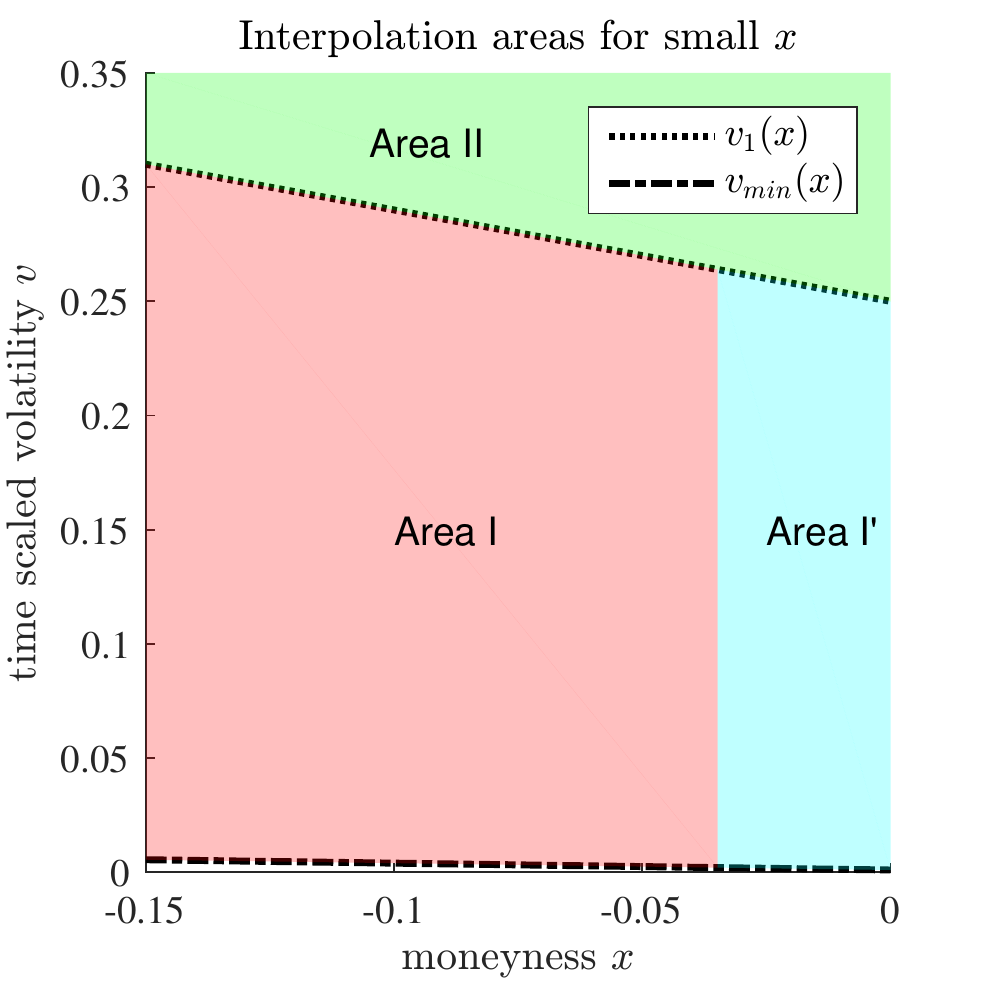}
    \end{minipage}
\captionof{figure}{The four different interpolation areas of the Chebyshev method.}\label{fig:PlotAreas}%
\end{figure}
Moreover, we replace the boundary call prices $c_1(x)$, $c_2(x)$ and $c_{max}(x)$ by univariate interpolations to reduce the runtime further. The evaluation of $c_{min}(x)$, however, is done directly, since for low volatilities the call price is hard to approximate. For this step we use the implementation of the call price function provided in \cite{Jaeckel2015}, which is of very high precision.

\subsection{Algorithmic structure} \label{section:algorithm}
Our method allows for an online/offline decomposition:
\begin{itemize}
	\item \textit{offline-phase} (preparation):\\
		In each area, we compute the implied volatilities on a $N\times N$ grid of Chebyshev points. Then we apply the \textit{chebfun2} algorithm with pre-specified accuracy and obtain a low-rank approximation.
	\item \textit{online-phase} (real-time evaluation):\\
		In the online phase implied volatilities are computed from real-time data, containing a vector of call prices $C\in\IR^n$ and the corresponding strikes $K\in\IR^n$, spot prices $S_0\in\IR^n$, maturities $T\in\IR^n$ and interest rates $r\in\IR^n$.
		\begin{itemize}
			\item \textit{Normalization}: We calculate the normalized call price $c$ and the forward moneyness $x$ from the data. Option prices with $x>0$ need to be transformed to prices with moneyness $-x$ by Formula \eqref{ITM_option_formula}.
			\item \textit{Splitting}: For each pair $(x,c)$, we need to find the corresponding area. As the computation of $c_{min}(x)$ requires the most computational effort, we proceed as follows. First, we compute $c_{max}(x)$ and check if $c\leq c_{max}(x)$. Next, we check if $c<c_{2}(x)$ and eventually $c<c_{1}(x)$. Only in the latter case, do we compute $c_{min}(x)$ and check whether $c\geq c_{min}(x)$.			
			\item \textit{Transformation}: We compute the transformed call prices $\phi_{i,x}(c)$ and moneyness $\phi_x(x)$ with the respective transformations. 
			\item \textit{Evaluation}: We evaluate the bivariate Chebyshev interpolations provided in the \textit{offline-phase} at the transformed call prices and moneyness to obtain the time-scaled implied volatility.
		\end{itemize}
\end{itemize}

The runtime of the \textit{online}-phase is primarily determined by the \textit{splitting} and the \textit{evalutation}-phase. 
The evaluation of the bivariate interpolations can be done in different ways and can be performed in very few computational steps depending on the required accuracy.\\
%
%

For optimal efficiency in the evaluation step, we consider a bivariate Chebyshev interpolation of a function $f(x,y)$ in the low rank form $I^{N_1,N_2}=\sum_{j=1}^{k}d_j c_j(y) r_j(x)$ where $r_j(x)$ and $c_j(y)$ are univariate Chebyshev interpolations of rank $N_1$ and $N_2$. More precisely,
\[
	r_j(x)=\sum_{i=0}^{N_1-1}a_i T_i(x) \text{ and } c_j(y)=\sum_{i=0}^{N_2-1}b_iT_i(y)
\]
The Chebyshev polynomials $T_0,T_1, ...,T_{N_1-1}$ can be computed in different ways, for instance by $T_k(x)=\cos(k\cos^{-1}(x))$ or by the iterative formula $T_0(x)=0$, $T_1(x)=1$, $T_{k+1}(x)=2xT_k(x)-T_{k-1}(x)$. It turns out that for large data sets the iterative evaluation of the Chebyshev polynomials is advantageous compared to the cosine formula as only simple additions and multiplications are involved while the evaluation of $\cos$ and $\cos^{-1}$ is slightly slower. Therefore we use this approach in our implementation. 

After setting up the Chebyshev method for a pre-specified accuracy we obtain a low-rank approximation for each of the four areas. Table \ref{table:ranks} displays the ranks $k$ and the grid sizes $N_{1},N_{2}$ of the low rank interpolation operator for the three specified accuracies $10^{-6}$ (low accuracy), $10^{-9}$ (medium accuracy) and $10^{-12}$ (high accuracy). As expected the ranks and grid sizes are higher for a higher accuracy. Moreover, we observe that we need more interpolation nodes in Area I and Area I' to obtain the same level of accuracy as in Area II and Area III.

\begin{center}
\scriptsize
\centering
\begin{tabular}{@{}*{6}{l}@{}}
\toprule
Area & low accuracy & medium accuracy & high accuracy\\
\midrule
Area I & $k=10$, $N_{1}=25$, $N_{2}=36$ & $k=16$, $N_{1}=46$, $N_{2}=79$ & $k=22$, $N_{1}=67$, $N_{2}=122$\\
Area I' & $k=9$, $N_{1}=27$, $N_{2}=18$ & $k=16$, $N_{1}=51$, $N_{2}=39$ & $k=23$, $N_{1}=77$, $N_{2}=57$\\
Area II & $k=6$, $N_{1}=21$, $N_{2}=20$ & $k=11$, $N_{1}=36$, $N_{2}=33$ & $k=14$, $N_{1}=51$, $N_{2}=47$\\
Area III & $k=5$, $N_{1}=11$, $N_{2}=9$ & $k=7$, $N_{1}=17$, $N_{2}=14$ & $k=9$, $N_{1}=23$, $N_{2}=19$\\
\bottomrule
\end{tabular}\label{table:ranks}
\captionof{table}{Rank $k$ and grid sizes $N_{1},N_{2}$ of the low rank Chebyshev interpolation in the different areas for three different levels of pre-specified accuracy.}
\end{center}

\section{Numerical Results}\label{section_numerik}
We compare our approximation method to
\begin{itemize}
\item the \cite{Jaeckel2015} method,
\item the approximation formula given in \cite{Li2008},
\item the approximation formula given in \cite{Li2008} with the proposed polishing of two Newton-Raphson iterations,
\item the Newton-Raphson algorithm with the starting point given in \cite{ManasterKoehler1982}. The algorithm terminates 
if $|v_n-v_{n-1}|<10^{-6}$.
\end{itemize}

In order to do so, we first choose a domain $\mathcal{D}_{1}$ on which all methods can be applied and compare the resulting errors and runtimes (Section \ref{ComparisonDomain1}). On the complete domain $\mathcal{D}_{2}$, we compare the proposed method to the \cite{Jaeckel2015} method and the Newton-Raphson algorithm as those are the only ones that can also be applied on this set (Section \ref{ComparisonDomain2}). Finally, we include actual market data (Section \ref{ComparisonMarketData}). All codes are written in \textit{Matlab} R2014a and the experiments are run on a computer with Intel Xeon CPU with 3.10 GHz with 20 MB SmartCache.

\subsection{Comparison on Domain $\mathcal{D}_{1}$}\label{ComparisonDomain1}
The domain on which all methods work is the domain of \cite{Li2008} bounded below by $v_{min}(x)$, i.e.
\[
\mathcal{D}_{1}:=\left\{-0.5\leq x\leq 0.5, 0\leq v\leq 1, \max\left(\frac{|x|}{2},v_{min}(-|x|)\right)\leq v\right\}
\]
See Figure \ref{fig:PlotArea_Li} for a comparison of the domain of \cite{Li2008} and the domain of the Chebyshev method.  On $\mathcal{D}_{1}$ we compute normalized call prices on a $1000\times 1000$-grid, where the distribution of the points is determined as in the numerical example of Section \ref{section_simple}. We compare the runtimes and errors in the time-scaled volatilities $\Delta v:=|v-v^{imp}|$ and the repricing errors $\Delta c:=|c(x,v)-c(x,v^{imp})|$ of the methods. Figure \ref{fig:ComparisonDomain1} illustrates the errors $\Delta v$ of the reference methods. Figure \ref{fig:ComparisonDomain1Cheb} displays the errors of the Chebyshev approach for three different pre-specified accuracies.

\begin{center}
\includegraphics[width=10cm]{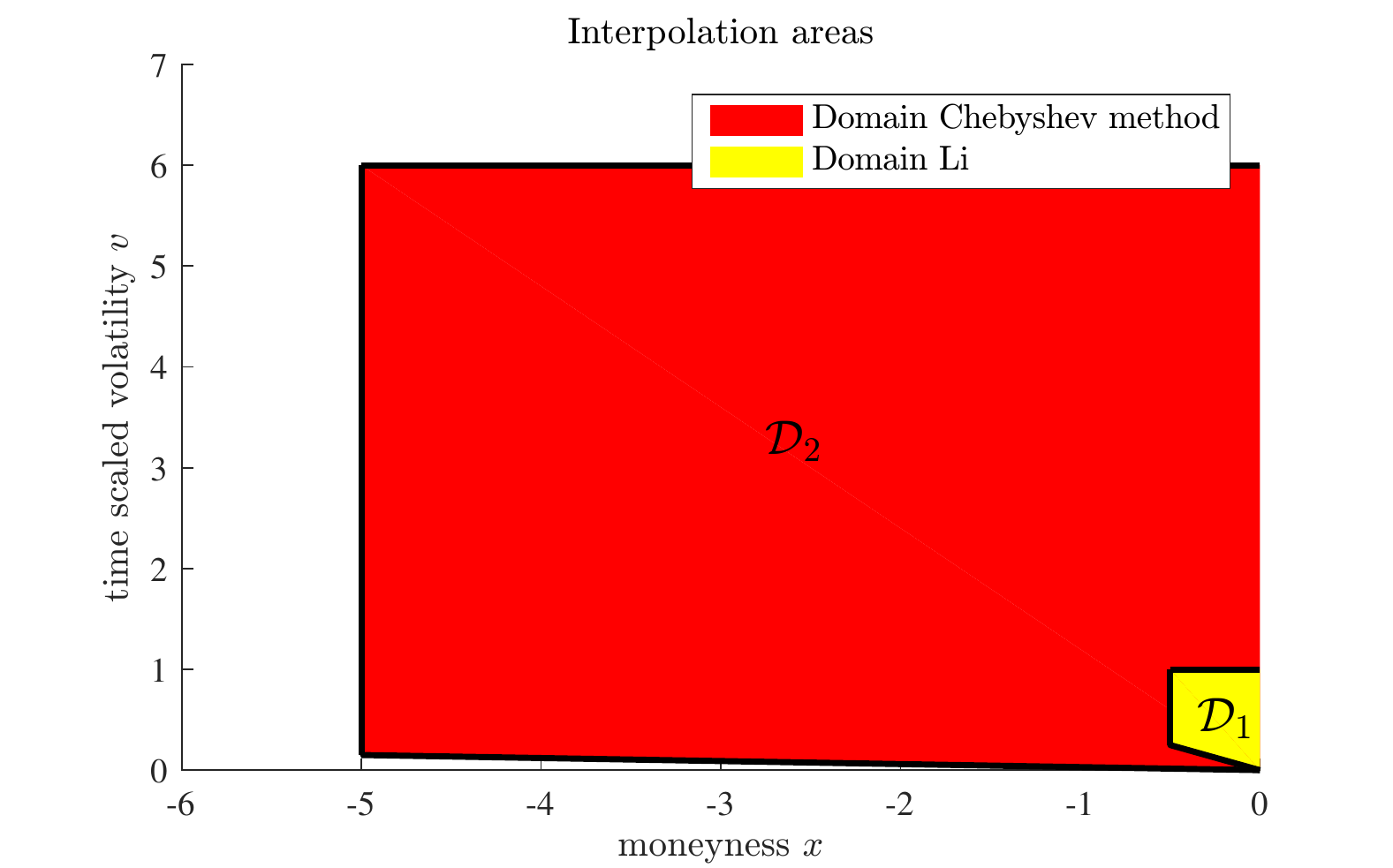}
\captionof{figure}{Domain $\mathcal{D}_{2}$ of the Chebyshev interpolation (red), domain of Li (yellow) and domain $\mathcal{D}_{1}$ as the intersection of both.}\label{fig:PlotArea_Li}%
\end{center}

The \cite{Jaeckel2015} method comes with a solution close to machine precision for all input parameters and thus qualifies as our reference method in the offline-phase of the Chebyshev approximation. Also the Newton-Raphson algorithm reaches very high precision. The approximation of \cite{Li2008}, however, is not able to reach the same range of precision. As Table \ref{table_domain1} shows, the mean error of $\sigma$ is a factor even $10^{10}$ higher than J\"ackel's approximation. The proposed modification of \cite{Li2008} with two additional Newton-Raphson steps reduces the error. However, for low volatilities the effect is rather small and the maximal error is still in the region of $10^{-5}$, see Table \ref{table_domain1}.

Figure \ref{fig:ComparisonDomain1Cheb} displays the interpolation error of the Chebyshev method for three different pre-specified accuracies. The error is of the same order for the whole interpolation domain, which shows that a pre-specified accuracy can be reached for all input parameters with the same complexity. 

\begin{center}
\includegraphics[width=12cm]{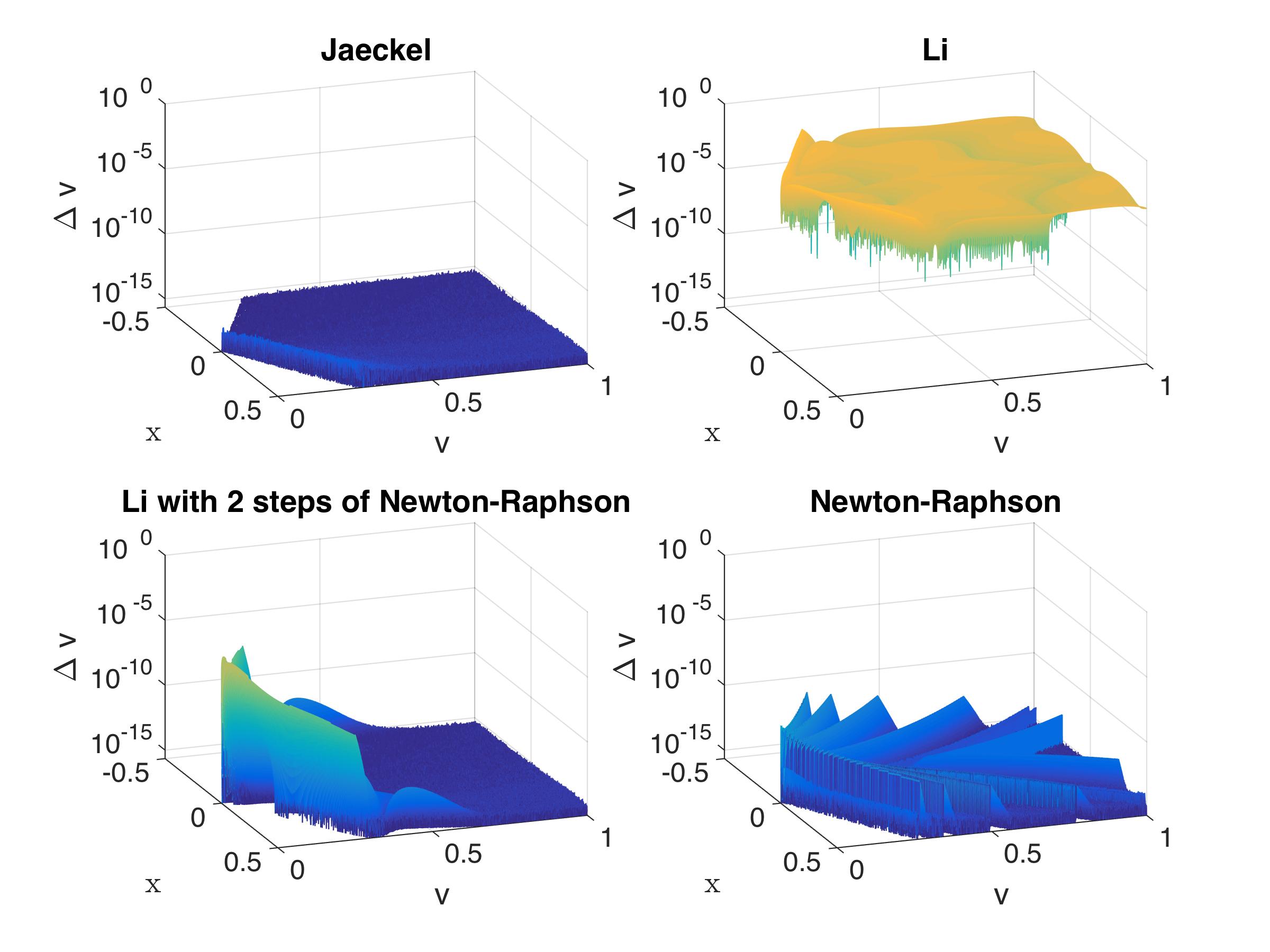}
\captionof{figure}{Errors $\Delta v:=|v-v^{imp}|$ of the reference methods}\label{fig:ComparisonDomain1}%
\end{center}

\begin{center}
\includegraphics[width=12cm]{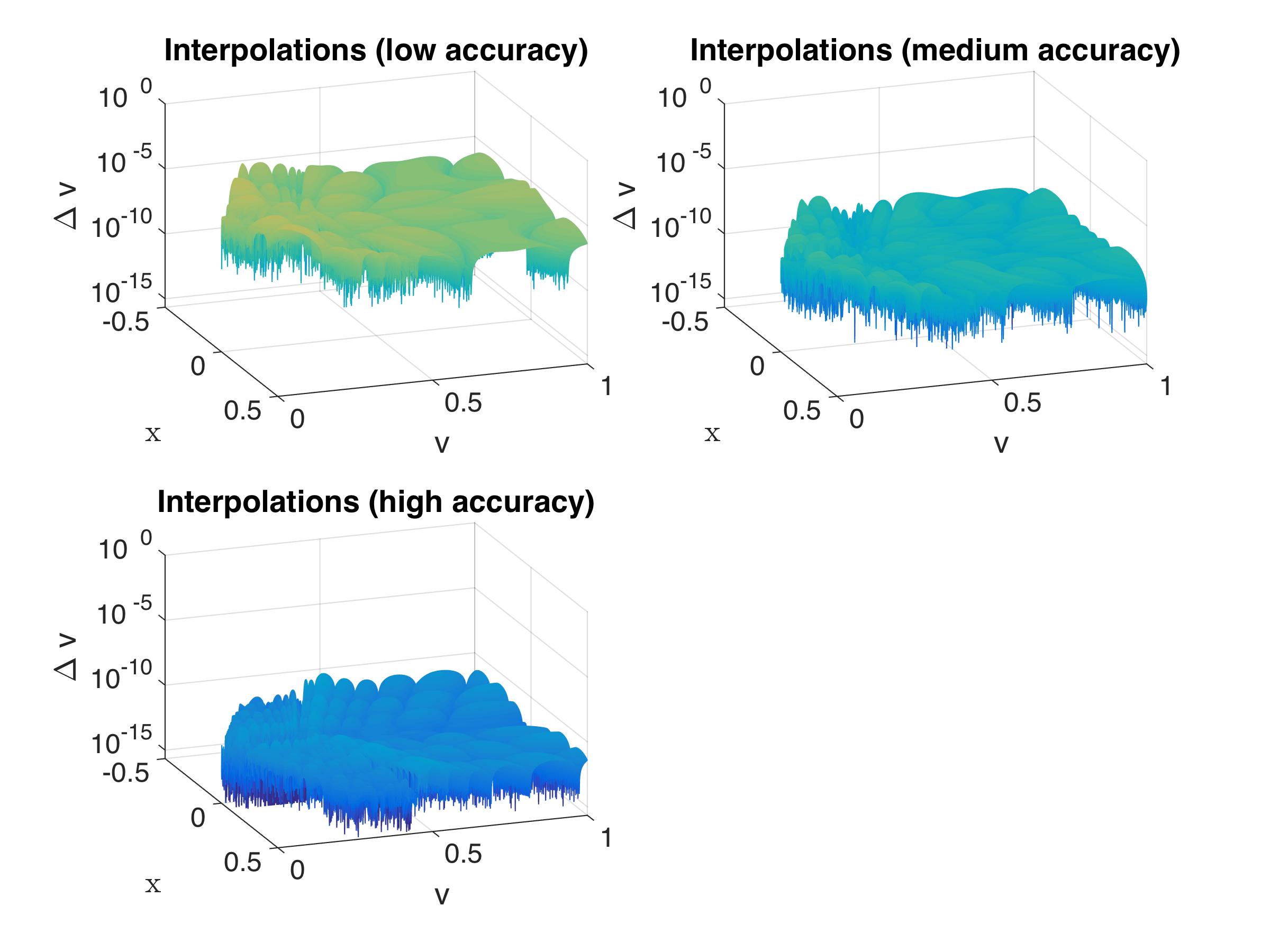}
\captionof{figure}{Errors $\Delta v:=|v-v^{imp}|$ of the Chebyshev approach with three different pre-specified accuracies.}\label{fig:ComparisonDomain1Cheb}%
\end{center}

Table \ref{table_domain1} shows the maximal and the mean error in terms of the time-scaled volatilities and the normalized call prices as well as the runtime as a proportion of the runtime of the Newton-Raphson method, which takes $1.45s$. For the Chebyshev method, the runtime measures the time of the online phase. When comparing the runtimes, the Li method is the fastest. It comes, however, with the lowest precision of a maximal error in $\sigma$ of $3.26\cdot10^{-3}$. For a higher precision in the range of $10^{-5}$, the Chebyshev method with low accuracy turns out to be faster than the improved Li method. Comparing the mean, the same holds for the Chebyshev method with medium accuracy. For very high precisions the Chebyshev method with high accuracy is faster than the Newton-Raphson approach. Compared to J\"ackel's method, the Chebyshev approach is two times faster but with a maximal error of $10^{-11}$ instead of $10^{-14}$.

\renewcommand{\arraystretch}{1.5}
 
\begin{center}
\scriptsize
\centering
\begin{tabular}{@{}*{6}{l}@{}}
\toprule
Method & max $|\Delta\sigma|$ & mean $|\Delta\sigma|$ & max $|\Delta c|$ & mean $|\Delta c|$ & runtime\\
\midrule
J\"ackel &$2.80\cdot10^{-14}$ &$4.57\cdot10^{-16}$ &$1.67\cdot10^{-15}$ &$9.99\cdot10^{-17}$ &$1.39$\\
Li &$3.26\cdot10^{-3}$ &$3.42\cdot10^{-4}$ &$2.15\cdot10^{-4}$ &$9.43\cdot10^{-5}$ &$0.12$\\
Li with 2 steps of Newton-Raphson &$2.02\cdot10^{-5}$ &$6.12\cdot10^{-9}$ &$1.10\cdot10^{-6}$ &$3.89\cdot10^{-10}$ &$0.63$\\
Newton-Raphson &$2.05\cdot10^{-10}$ &$6.32\cdot10^{-14}$ &$2.91\cdot10^{-11}$ &$1.00\cdot10^{-14}$ &$1$\\
Chebyshev method (low accuracy) &$1.52\cdot10^{-5}$ &$1.40\cdot10^{-6}$ &$4.91\cdot10^{-6}$ &$3.94\cdot10^{-7}$ &$0.40$\\
Chebyshev method (medium accuracy) &$3.20\cdot10^{-8}$ &$2.17\cdot10^{-9}$ &$3.52\cdot10^{-9}$ &$5.92\cdot10^{-10}$ &$0.55$\\ 
Chebyshev method (high accuracy) &$4.88\cdot10^{-11}$ &$4.78\cdot10^{-12}$ &$1.51\cdot10^{-11}$ &$1.41\cdot10^{-12}$ &$0.67$\\
\bottomrule
\end{tabular}\label{table_domain1}
\captionof{table}{Interpolation error and runtimes on domain $\mathcal{D}_{1}$.}
\end{center}

\subsection{Comparison on Domain $\mathcal{D}_{2}$}\label{ComparisonDomain2}
We compare the Chebyshev method on the large domain $\mathcal{D}_{2}$ to the Newton-Raphson approach and the algorithm of J\"ackel. The errors and runtimes on a $1000\times1000$ grid, specified as in Section \ref{ComparisonDomain1}, are computed. Figure \ref{fig:ComparisonDomain2} and \ref{fig:ComparisonDomain2Cheb} illustrate the resulting errors of the reference methods and the Chebyshev approach. The observations of the error behaviour on the larger domain $\mathcal{D}_{2}$ are consistent with that on the smaller domain $\mathcal{D}_{1}$, see Figure \ref{fig:ComparisonDomain1} and Figure \ref{fig:ComparisonDomain1Cheb}.

\begin{center}
\includegraphics[width=12cm]{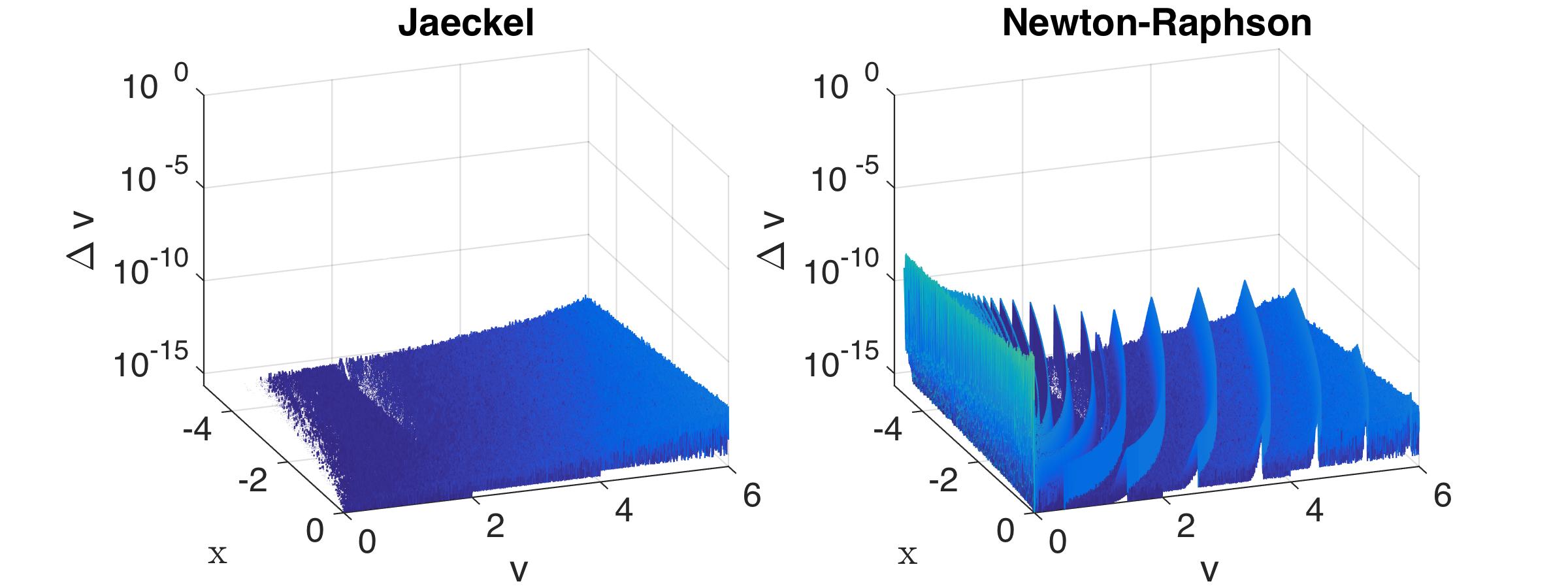}
\captionof{figure}{Errors $\Delta v:=|v-v^{imp}|$ of the reference methods}\label{fig:ComparisonDomain2}%
\end{center}

\begin{center}
\includegraphics[width=12cm]{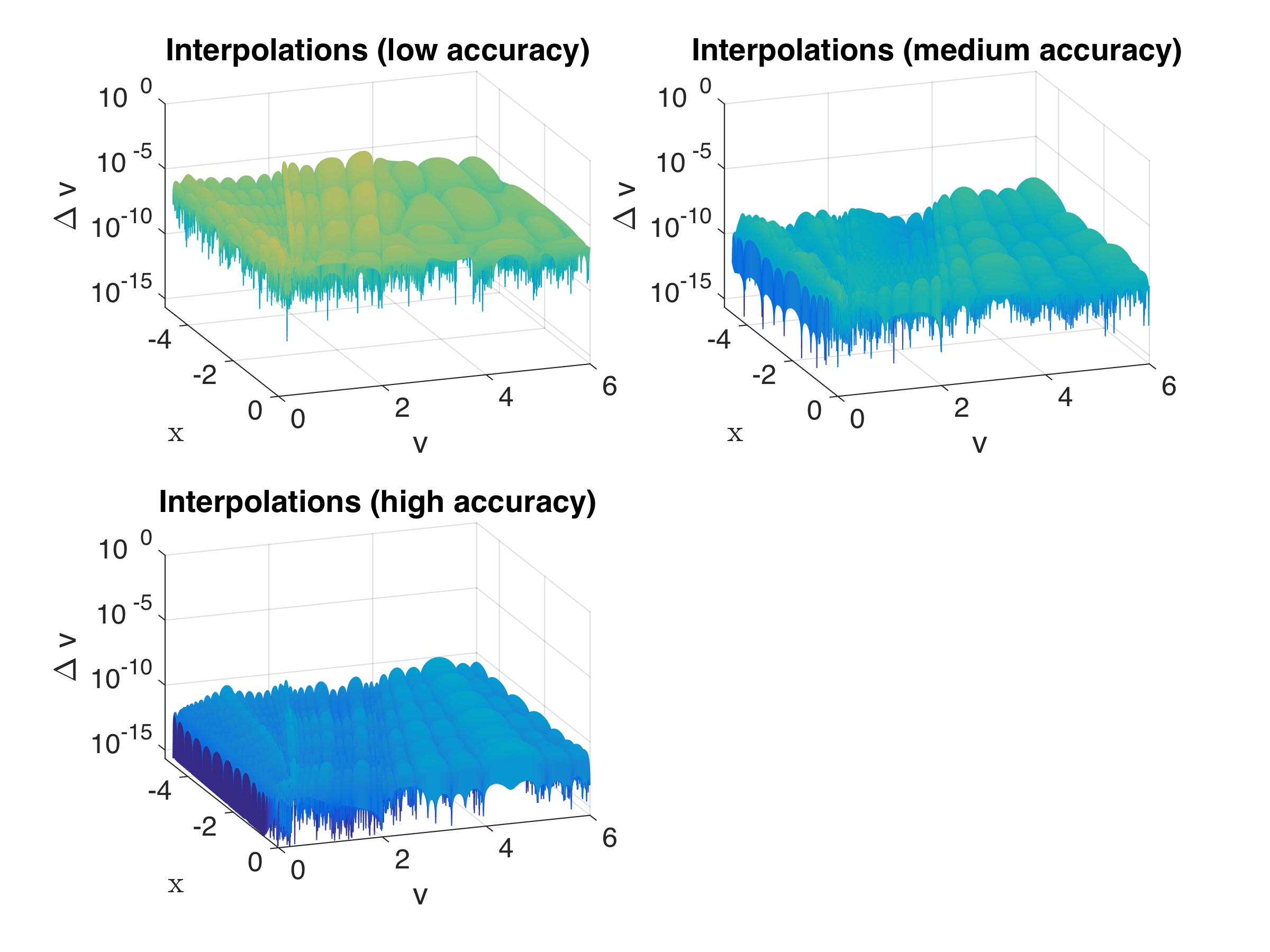}
\captionof{figure}{Errors $\Delta v:=|v-v^{imp}|$ of the reference methods}\label{fig:ComparisonDomain2Cheb}%
\end{center}


Table \ref{table_domain2} shows the maximal and the mean error as well as the runtimes scaled as in \ref{table_domain1}. Here, the Newton-Raphson method takes $4.29s$. 
\begin{center}
\scriptsize
\centering
\begin{tabular}{@{}*{6}{l}@{}}
\toprule
Method & max $|\Delta\sigma|$ & mean $|\Delta\sigma|$ & max $|\Delta c|$ & mean $|\Delta c|$ & runtime\\
\midrule
J\"ackel &$5.30\cdot10^{-13}$ &$5.35\cdot10^{-15}$ &$2.55\cdot10^{-15}$ &$7.10\cdot10^{-17}$ &$0.52$\\
Newton-Raphson &$8.34\cdot10^{-8}$ &$6.64\cdot10^{-12}$ &$1.94\cdot10^{-11}$ &$1.28\cdot10^{-15}$ &$1$\\ 
Chebyshev method (low accuracy) &$2.55\cdot10^{-5}$ &$1.85\cdot10^{-6}$ &$4.63\cdot10^{-6}$ &$1.42\cdot10^{-7}$ &$0.14$\\ 
Chebyshev method (medium accuracy) &$4.42\cdot10^{-8}$ &$2.38\cdot10^{-9}$ &$4.02\cdot10^{-9}$ &$1.36\cdot10^{-10}$ &$0.16$\\ 
Chebyshev method (high accuracy) &$1.66\cdot10^{-10}$ &$1.32\cdot10^{-11}$ &$1.52\cdot10^{-11}$ &$4.83\cdot10^{-13}$ &$0.20$\\
\bottomrule
\end{tabular}\label{table_domain2}
\captionof{table}{Interpolation error and runtimes on domain $\mathcal{D}_{2}$.}
\end{center}

To reach a medium accuracy in the maximal error in the range of $10^{-8}$, the Chebyshev method is more than six times faster than the Newton-Raphson approach. Moreover, the Chebyshev method is able to reach higher accuracies of $10^{-10}$ and still needs only $20\%$ of the runtime of Newton-Raphson. J\"ackel's method reaches very high precisions and is faster than Newton-Raphson. Compared the J\"ackel method, the Chebyshev method allows us to pre-specify accuracies and reduce the runtimes significantly. For example, if accuracies in the region of $10^{-8}$ are sufficient, the Chebyshev method is more than three times faster than J\"ackel's approach.

\subsection{Comparison for market data}\label{ComparisonMarketData}
In Section \ref{section_data_domain} we investigated market data of options and concluded that a significant part of the options is not covered by the domain of \cite{Li2008}. This was the motivation to consider a much larger interpolation domain for the Chebyshev method. An empirical investigation confirms that all the options shown in Figure \ref{fig:Domains_Data} lie within our domain. 

Next, we compare the Chebyshev method on this market data to the Newton-Raphson approach and the algorithm of J\"ackel. The errors and runtimes are computed for options on the S\&P 500 index traded on 7/17/2017 (Source Thomson Reuters Eikon). We use the same options as for Figure \ref{fig:Domains_Data}. To obtain more reliable results for the runtime comparison we compute the implied volatilities of the options $5000$ times. 

Table \ref{table_marketdata} shows the maximal and the mean error as well as the runtimes scaled as in Section \ref{table_domain1}. Here, the Newton-Raphson method takes $5.72s$. 
\begin{center}
\scriptsize
\centering
\begin{tabular}{@{}*{6}{l}@{}}
\toprule
Method & max $|\Delta\sigma|$ & mean $|\Delta\sigma|$ & max $|\Delta c|$ & mean $|\Delta c|$ & runtime\\
\midrule
J\"ackel &$8.05\cdot10^{-16}$ &$1.40\cdot10^{-16}$ &$2.11\cdot10^{-15}$ &$2.43\cdot10^{-16}$ &$0.89$\\
Newton-Raphson &$1.78\cdot10^{-10}$ &$2.91\cdot10^{-12}$ &$7.72\cdot10^{-12}$ &$2.22\cdot10^{-13}$ &$1$\\ 
Chebyshev method (low accuracy) &$1.57\cdot10^{-5}$ &$2.95\cdot10^{-6}$ &$4.44\cdot10^{-6}$ &$4.78\cdot10^{-7}$ &$0.37$\\ 
Chebyshev method (medium accuracy) &$4.19\cdot10^{-8}$ &$3.87\cdot10^{-9}$ &$3.45\cdot10^{-9}$ &$3.98\cdot10^{-10}$ &$0.48$\\ 
Chebyshev method (high accuracy) &$1.73\cdot10^{-11}$ &$2.21\cdot10^{-12}$ &$2.70\cdot10^{-12}$ &$2.91\cdot10^{-13}$ &$0.58$\\
\bottomrule
\end{tabular}\label{table_marketdata}
\captionof{table}{Interpolation error and runtimes for S\&P 500 market data.}
\end{center}
The results are similar to those of Section \ref{ComparisonDomain2}. The Chebyshev method is the fastest of the three methods and reaches the target accuracies. The method is about twice as fast as the Newton-Raphson approach for similar accuracies. Again, J\"ackel's method reaches very high precisions but it is significantly slower than the Chebyshev method.

Besides the observed gain in efficiency the Chebyshev method enjoys conceptual advantages. It delivers a closed-form approximation in a simple polynomial structure. The code is easy to implement and maintain. Moreover, the proposed approach can be applied to other problems of similar structure. The following section illustrates this flexibility.

\section{Laplace implied volatility}
Besides the Black-Scholes implied volatility there are several other models with implied volatilities. To overcome the problems of thin tails in the Black-Scholes model, \cite{Madan2016} proposes the replacement of the density of the normal distribution with a Laplace density. This leads to a model with fatter tails without adding additional parameters. The stock price process in this model is defined by
\begin{align}\label{Laplace_stock_price}
S_{t}=S_{0}\exp\left((r-q)t + X_{t} + \log\left(1-\frac{\sigma^{2}t}{2}\right)\right),
\end{align}
where $X_{t}$ is distributed according to the time-dependent Laplace density
\begin{align}\label{Laplace_density}
g(x,t)=\frac{1}{\sigma\sqrt{2t}}e^{-\frac{\sqrt{2}|x|}{\sigma\sqrt{t}}}, \qquad x\in\mathbb{R}.
\end{align}
The call price in the model is given by
\begin{align*}
C&(S_0,K,r,q,t)=\\
&\begin{cases}
e^{-qt}S_{0}\frac{e^{-(\sqrt{2}-\sigma\sqrt{t})|d|}}{2}\left(1+\sigma\sqrt{t/2}\right) - e^{-rt}K\frac{e^{-\sqrt{2}|d|}}{2}, & d>0\\
Ke^{-rT}\left(\frac{e^{-\sqrt{2}|d|}}{2}-1\right)-S_0e^{-qt}\left(\frac{e^{-(\sqrt{2}+\sigma\sqrt{t})|d|}}{2}\left(1-\sigma\sqrt{t/2}\right)-1\right),& d<0
\end{cases}
\end{align*}
with
\begin{align*}
d=\frac{\log(K/S_{0})}{\sigma\sqrt{t}}-\frac{(r-q)\sqrt{t}}{\sigma}-\frac{\log(1-\frac{\sigma^{2}t}{2})}{\sigma\sqrt{t}}.
\end{align*}
\cite{Madan2016} shows that this model can be used for hedging purposes and outperforms classical delta hedging in the Black-Scholes model. \cite{MadanWang2016} considered the application of the model to risk management. For both, hedging and risk management, it is necessary to have a fast and accurate formula for the Laplace implied volatility. To this end, we apply the bivariate Chebyshev method to implied volatilities based on the Laplace density.

As in the previous case, we normalize the call price by setting $v=\sigma\sqrt{T}$, $x=\log(S_0e^{(r-q)T}/K)$ and $C(S_0,K,r,q,t)=\sqrt{S_0e^{-(r+q)T}K}c(x,v)$ to 

\begin{align}\label{Laplace_call_price_normalized}
c(x,v)=
\begin{cases}
\frac{e^{-(\sqrt{2}-v)|d|+x/2}}{2}(1+v/\sqrt{2})-\frac{e^{\sqrt{2}|d|-x/2}}{2}, & d>0\\
e^{-x/2}\left(\frac{e^{-\sqrt{2}|d|}}{2}-1\right)-e^{x/2}\left(\frac{e^{-(\sqrt{2}+v)|d|}}{2}\left(1-v/\sqrt{2}\right)-1\right),& d<0
\end{cases}
\end{align}
with
\begin{align*}
d=-\frac{x}{v}-\frac{\log{(1-\frac{v^2}{2})}}{v}.
\end{align*}

Similar to Section \ref{section_norm_bs_price}, we thus have reduced the approximation problem to a bivariate interpolation. For the domain $0.25\leq v\leq1$, $-0.4\leq x\leq0$, we perform a bivariate Chebyshev interpolation of the Laplace implied volatility. At the interpolation nodes a Brent-Dekker algorithm is used to compute the implied volatilities. Figure \ref{fig:TestLaplace_OutOfMoney} shows the exponential error decay of the interpolation on a $N\times N$-Chebyshev grid. The interpolation is already in the region of $10^{-11}$ for $N=50$. This shows the high potential of the method in the Laplace model, comparable to the numerical example in Section \ref{section_simple}. In order to obtain high efficiency on a larger domain, one can establish a splitting procedure with appropriate scaling functions by exploiting the limit behaviour of the Laplace call price, in the spirit of Section \ref{section_splitting}. 

\begin{figure}[htb]
    \centering
    \begin{minipage}{0.45\linewidth}
        \centering
        \includegraphics[width=7cm]{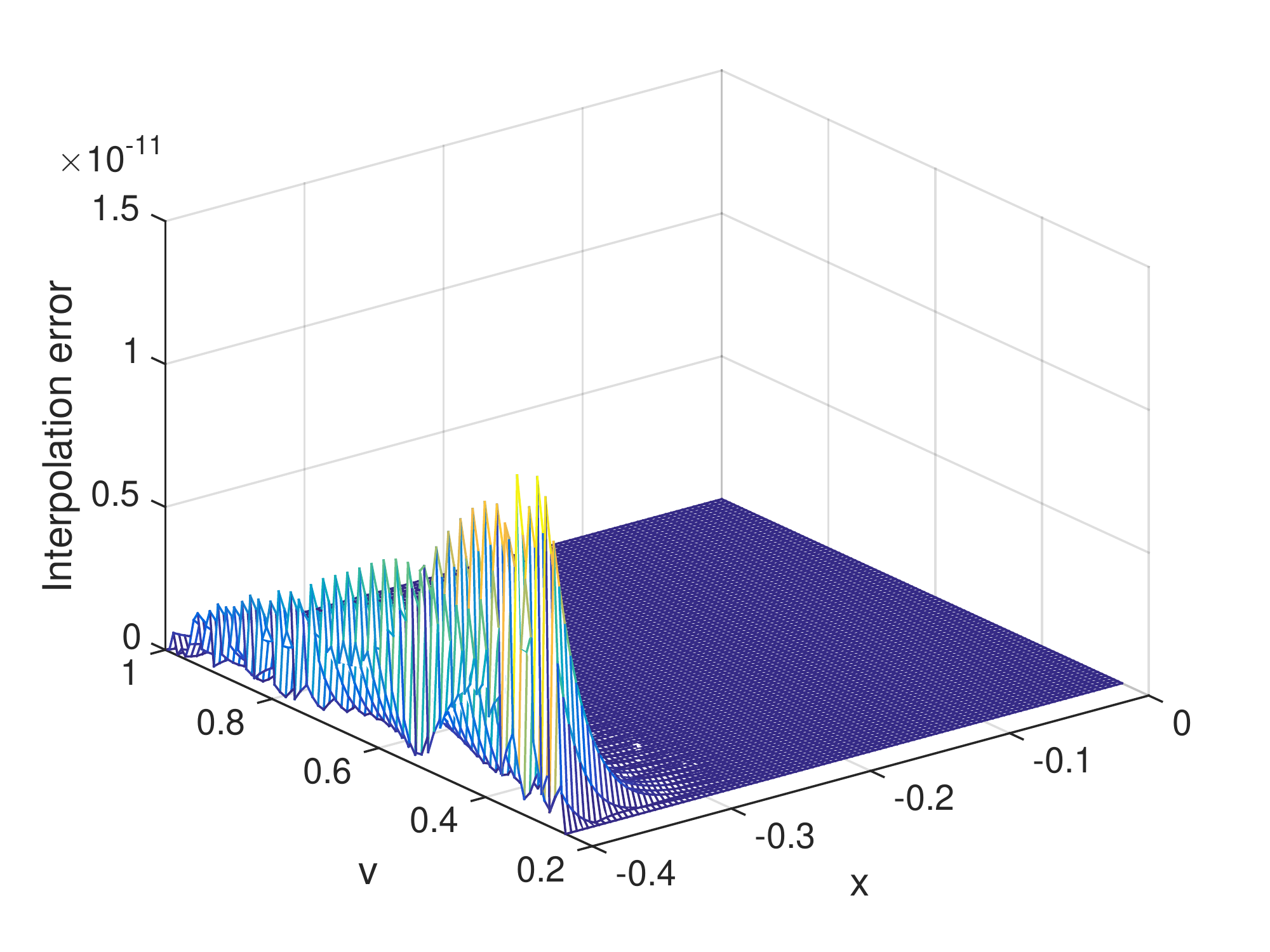}
    \end{minipage}
    \begin{minipage}{0.45\linewidth}
        \centering
        \includegraphics[width=7cm]{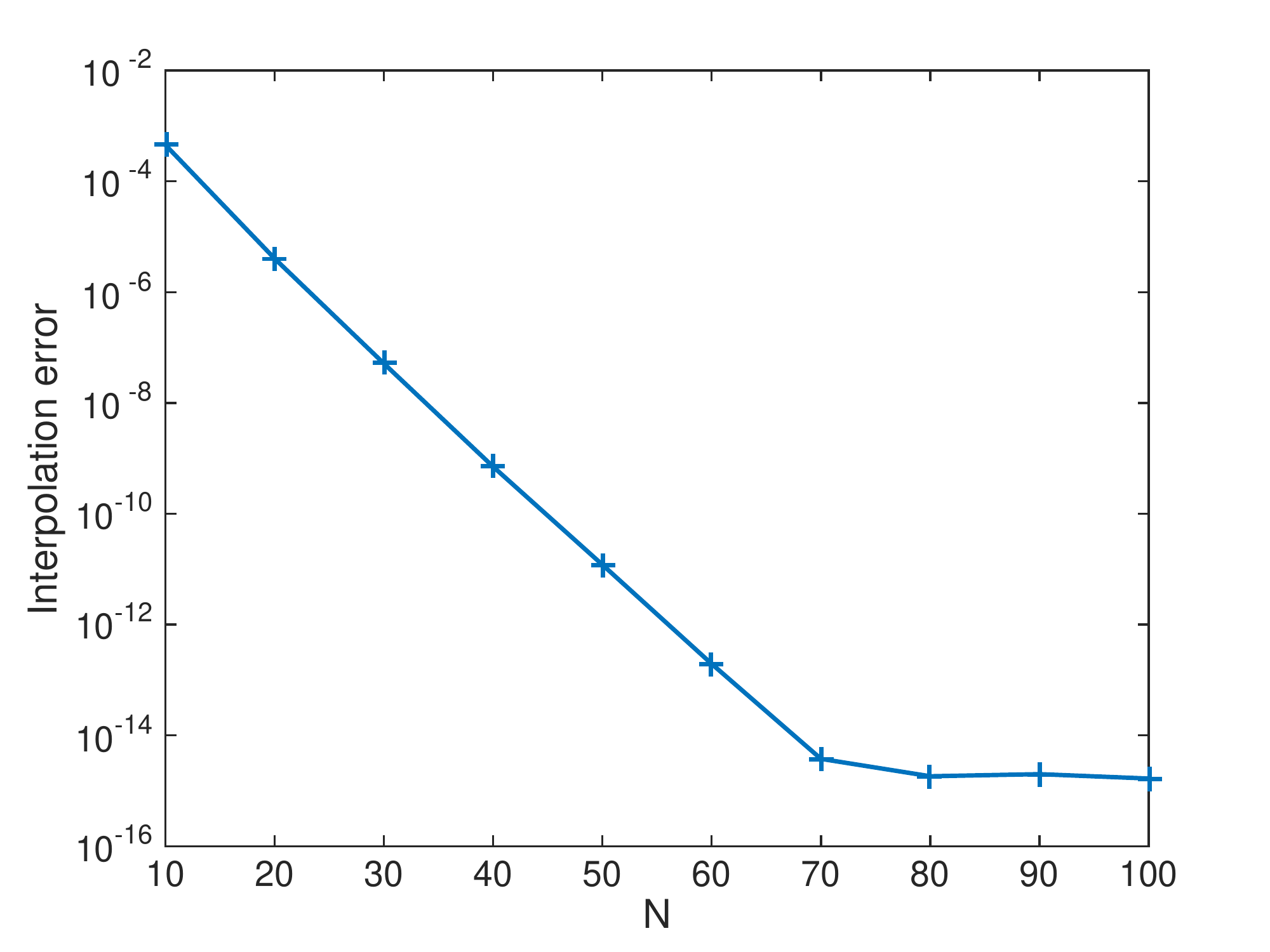}
    \end{minipage}
\captionof{figure}{Exponential error decay of the bivariate interpolation of the Laplace implied volatility on a $N\times N$-Chebyshev grid for $0.25\leq v\leq1$, $-0.4\leq x\leq0$. For $N=50$ the error is already in the region of $10^{-11}$.}\label{fig:TestLaplace_OutOfMoney}%
\end{figure}

%


\section{Conclusion}
We have introduced a new approximation method to compute the implied volatility. The backbone of the method is a bivariate Chebyshev interpolation. We have set up an interpolation domain, which is able to cover all relevant options based on observed market data. In order to achieve highest efficiency, we have split the domain into different interpolation areas with appropriate scaling functions. A theoretical error analysis shows subexponential convergence and a combination with low-rank techniques allows us to enhance the observed efficiency. Compared to other non-iterative approximation methods, the Chebyshev method is able to cover all relevant option data, including deep in and far out of the money options as well as low and high volatilities, see Figure \ref{fig:Domains_Data} and Figure \ref{fig:PlotArea_Li}. Moreover, numerical experiments show that the Chebyshev method achieves considerably higher accuracies on the common domain $\mathcal{D}_{1}$. In comparison to the iterative method of \cite{Jaeckel2015}, the Chebyshev method can reduce the runtimes significantly by pre-specifying the target accuracy. Besides the gain in efficiency, the Chebyshev method exhibits conceptual advantages:
\begin{itemize}
	\item \textit{Closed form bivariate approximation formula}: The Chebyshev interpolations in all areas have the \textit{polynomial structure}
		\begin{align}\label{approximation_formula_impvol}
		v(c,x)\approx\sum_{j=1}^k d_j c_j(\phi_i(c)) r_j(\phi_x(x))
		\end{align}
		where $\phi_i$ and $\phi_x$ are the transformations on the respective area. This structure can be further explored to express derivatives in a simple form. For example the first derivative with respect to the call price is given by 
		\[
		\frac{\partial}{\partial c } v(c,x) =\left(\frac{\partial}{\partial v} c(v,x)\right)^{-1}=\sum_{j=1}^k d_j r_j(\phi_x(x))\cdot\frac{\partial}{\partial \phi_i(c)}c_j(\phi_i(c))\cdot\frac{\partial}{\partial c}\phi_i(c) 
		\]
We observe that the approximate derivative is again a function of $x$ and $c$ in polynomial structure. In particular, this avoids the computation of the implied volatility itself.
	\item \textit{Easy Implementation}: Once the interpolation operator is set up in an offline phase, the polynomial structure of the approximation formula \ref{approximation_formula_impvol} leads to simple code. This facilitates the transfer of the code to other systems and programming languages as part of the maintenance.
	\item \textit{Adaptability}: The efficiency of the Chebyshev method can be even further improved by incorporating additional knowledge. If the option data of interest lies in a domain smaller than $\mathcal{D}_{2}$, one can tailor the method to this domain by modifying the splitting.
 \end{itemize}

The Chebyshev method enjoys high flexibility and the approach can be transferred to similar problems. We have illustrated this by approximating the Laplace implied volatility.

\bibliographystyle{chicago}
  \bibliography{Diss_Literature}

\end{document}